\documentclass[aps,pre,twocolumn,groupedaddress,floatfix]{revtex4-2}

\usepackage{amsmath,amssymb,amsthm,braket,listings,physics,
            here,subfigure,longtable,bbm,mathtools,xcolor,bm}
\usepackage[normalem]{ulem}
\usepackage{graphicx}
\usepackage[hidelinks]{hyperref}

\usepackage{autobreak}

\newtheorem{theorem}{Theorem}
\newtheorem{proposition}{Proposition}

\newtheorem{lemma}{Lemma}

\newcommand{\pt}{\partial}
\newcommand{\relmiddle}[1]{\mathrel{}\middle#1\mathrel{}}

\DeclareMathOperator{\argmax}{argmax}

\usepackage{ulem}
\usepackage{xcolor}
\usepackage {cancel}
\usepackage{CJKutf8}
\usepackage{ifthen}
\newcommand{\COMM}[2]{{
\begin{CJK}{UTF8}{ipxm}
\ifthenelse{\equal{#1}{TJK}}{\color{red}}{
\ifthenelse{\equal{#1}{SN}}{\color{blue}}{
\ifthenelse{\equal{#1}{MK}}{\color{cyan}}{
\ifthenelse{\equal{#1}{BB}}{\color{magenta}}}}}
[#1: #2]
\end{CJK}
}}

\makeatletter
\renewcommand\appendix{\par
   \setcounter{section}{0}%
   \renewcommand{\thesection}{\Alph{section}}
   \renewcommand{\thesubsection}{\alph{subsection}}
}
\makeatother

\begin{document}


\title{Optimality theory of stigmergic collective information processing by chemotactic cells}


\author{Masaki Kato${}^{1}$}
\email{m\_kato@sat.t.u-tokyo.ac.jp}
\author{Tetsuya J. Kobayashi${}^{2,1,3}$}
\email{tetsuya@mail.crmind.net}
\affiliation{${}^{1}$Department of Mathematical Informatics, the Graduate school of Information Science and Technology, the University of Tokyo, Tokyo, Japan}
\affiliation{${}^{2}$Institute of Industrial Science, The University of Tokyo, Tokyo 153-8505, Japan}
\affiliation{${}^{3}$Universal Biology Institute, The University of Tokyo, Tokyo 113-8654, Japan}


\date{\today}

\begin{abstract}
Collective information processing is fundamental in various biological systems, where the cooperation of multiple cells results in complex functions beyond individual capabilities.
A distinctive example is collective exploration where chemotactic cells not only sense the gradient of guiding exogeneous cues originating from targets but also generate and modulate endogenous cues to coordinate their collective behaviors. 
While the optimality of gradient sensing has been studied extensively in the context of single-cell information processing, the optimality of collective information processing that includes both gradient sensing and gradient generation remains underexplored.
In this study, we formulate the collective exploration problem as a reinforcement learning (RL) by a population.
Based on RL theory, we derive the optimal exploration dynamics of agents and identify their structural correspondence with the Keller-Segel model, the established phenomenological model of collective cellular dynamics. 
Our theory identifies an optimal coupling relation between gradient sensing and gradient generation and demonstrates that the optimal way to generate a gradient qualitatively differs depending on whether the gradient sensing is logarithmic or linear.
The underlying RL structure is leveraged to compare the derived collective dynamics with single-agent searching dynamics, showing that distributed information processing by population enables a fraction of agents to reach the target robustly.
Our formulation provides a foundation for understanding the collective information processing mediated by dynamic sensing and modulation of cues.

\end{abstract}


\maketitle


\renewcommand{\theequation}{\arabic{equation}}

\section{Introduction}
\begin{figure}[tbp]
    \hfill
    \begin{minipage}[b]{\linewidth}
        \centering
    \includegraphics[width=\textwidth]{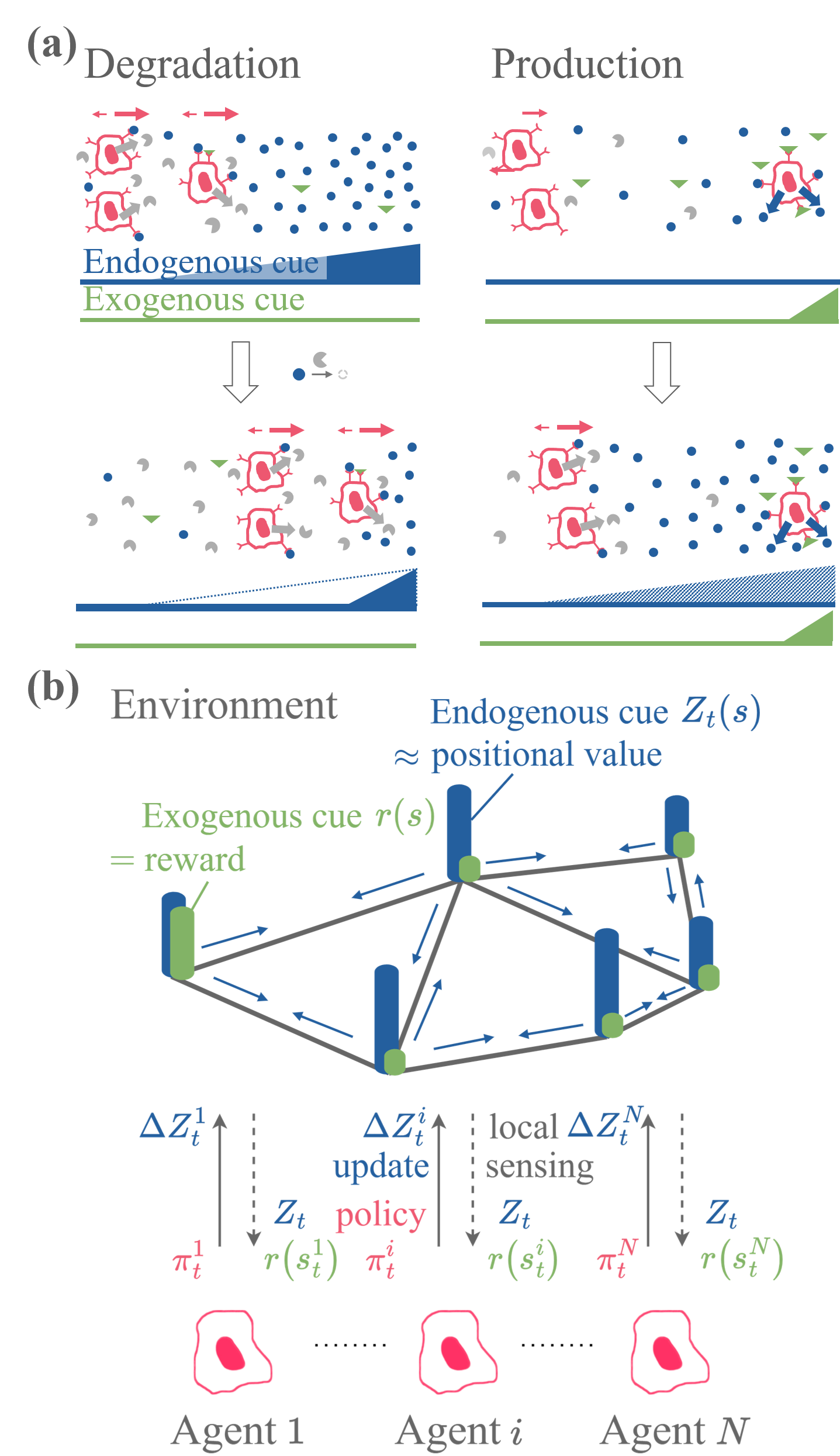}
    \end{minipage}
    \hfill
    \caption{(a) Schematics of an endogenous cue (attractant) degradation and production by a migrating cell population. The degradation by enzymes creates a gradient of the cue, guiding cells toward unexplored areas. 
    The production of endogenous cue in response to sensing the exogenous cue forms a gradient of the endogenous cue toward the source of the exogenous cue, recruiting cells to the source.
    (b) A schematic diagram of an agent population, which distributedly learns optimal migration toward the target locations with higher exogenous cue concentration (reward).
    The optimal state value $V^*$ for optimal migration is determined by the time-invariant concentration $r$ of exogenous cue, which contain the positional information of the targets in the environment.
    The agents learn $V^*$ by associating endogenous cue concentration $Z_t$ with their provisional estimation $V_t$ of $V^*$ ($V_t=\mathcal{F}[Z_t] \approx V^{*}$).
    At time $t$, each agent $i\in\{1,\ldots,N\}$ senses surrounding $Z_t$ and local $r(s_t^i)$ at each time step.
    Based on the sensed $r$ and $Z_t$, each agent transits between vertices according to its policy $\pi^i_t$ and individually modulates local $Z_t(s_t^i)$ by $\Delta Z^i_t=\Delta Z_t(s_t^i)$, where $\Delta Z^i_t<0$ corresponds to degradation, and $\Delta Z^i_t>0$ corresponds to production.  
    }
    \label{fig:landscape}
\end{figure}

Collective information processing by populations of cells is a widespread phenomenon throughout biology---from bacterial populations~\cite{Shapiro1998AnnuRevMicrobiol,Popat2015JRSocInterface,Flemming2016NatRevMicroBiol} to developing tissues~\cite{Friedl2009NatRevMolCellBiol,Theveneau2012DevBiol,Ellison2016PNAS}, immune systems~\cite{Luster2005NatImmu,Mantovani2011NatRevImmunol}, and metastatic tumors~\cite{Axelrod2006PNAS,Deisboeck2009Bioessays,Friedl2009NatRevMolCellBiol}--- where the cooperation of multiple cells leads to complex biological functions beyond individual capabilities. 
Among others, cue-mediated collective exploration is a representative class of information processing.
To share positional information about targets such as invaded pathogens or damaged cells in tissue, cells not only individually sense the chemical signal produced by targets, but also actively degrade and produce the extracellular field of chemical cues (Figure~\ref{fig:landscape} (a)). 
Active cue degradation enables a population of cells to efficiently migrate toward unexplored areas~\cite{Insall2022TrdCellBiol}. 
In addition, the production of cues by cells that have reached the target can recruit other cells and enhance the accumulation at the target, even if the target leaves only a faint signal into the environment~\cite{Oliveira2016NatRevImm}.
As a combination of sensing and modulation of cues, a population of \textit{Dictyostelium discoideum} can cooperatively find the goal of complex mazes~\cite{Tweedy2020Science}. 
Similar strategies may be employed by eukaryotic cell populations in tissues~\cite{Insall2022TrdCellBiol}, including immune~\cite{Breart2011JEM,Lammermann2013Nat} and cancer cells~\cite{Scherber2012IntegrBiol,Susanto2017JCS}.

Despite such accumulating experimental insights into collective information processing mediated by external chemical cues, 
we still lack the theoretical basis to investigate its functionality and efficiency.
Recent studies on cellular information processing of external cues have mainly focused on single-cell chemotaxis, particularly bacterial responses to ligand gradients.
Their optimality has been verified~\cite{Clark2005PNAS,Celani2010PNAS,Mattingly2021NatPhys,Nakamura2022PRR} by comparing optimal responses and behaviors derived theoretically with those measured and estimated by experiments~\cite{Segall1986PNAS,Kalinin2009BioPhys,Lazova2011PNAS,Tu2013AnnRevBioPhys}.
Nonetheless, extending those studies on single-cell responses to collective behaviors remains limited.
Existing works on collective cellular behaviors investigated only phenomenological aspects of the phenomena: cell population motility~\cite{KELLER1971225,Tindall2008BullMathBiol,Painter2019JTB,Alert2022PRL,Phan2024PNAS},  impacts of individual traits on population behavior~\cite{Fu2018NatCom,Mattingly2022PNAS}, and the interplay of population behaviors with growth~\cite{Fraebel2017Elife,Liu2019Nature,Cremer2019Nat}.
Only a few pioneering work~\cite{Pezzotta2018PRE} addressed the functionality and efficiency of collective behaviors.

Considering the efficiency of collective chemotaxis, there are two problems that require insights beyond single-cell studies.
One is the interdependence between the generation of concentration field of cues and the responses to the field for efficient collective chemotaxis.
Focusing only on the responses, as in the case of single-cell studies, is insufficient to understand the efficiency of collective chemotaxis.
For optimal migration and exploration, the way to shape the field should be consistent with the way to respond to it, as already suggested in the context of traveling wave existence~\cite{Horstmann2004JNonlinearSci}. 
A recent study addressed the optimality of cue-induced collective search using optimal control theory~\cite{Pezzotta2018PRE}.
They derived an optimal gradient sensing to a collectively generated cue field, yet the generation process of the field was incorporated heuristically.
Thus, the consistency problem remains unresolved.

The other problem is the benefits and drawbacks of being a collective over being an individual for solving information processing tasks.
Some organisms, such as human beings, possess high computational capabilities as individual agents~\cite{Behrens2018Neuron} whereas others have low capabilities yet work collectively. 
Both of them have been maintained evolutionarily.
Understanding the difference between individual and collective information processing is fundamental to comprehend the evolvability of biological computation~\cite{Sole2019RSocB,Navas2022TrendsCogSci}.

To address these problems in this study, we formulate a cell population exploring the environment through degradation and production of a cue as an agent population performing reinforcement learning (RL) in a distributed way.
This formulation allows us to clarify the optimal coupling between sensing and modulation of the cue and the functionality of distributed computation.
Specifically, we associate the state value $V_{t}$ in RL, the likelihood of accumulation at the target, with the cue concentration $Z_{t}$ in the environment and show that the process by which individual cells cooperatively modulate the cue based on their local observations can be interpreted as a collective learning of the optimal state value by agents based on their experiences (Figure~\ref{fig:landscape} (b)).
The derived dynamics are structurally equivalent to the Keller-Segel model, a widely employed mechanical model for explaining experiments of collective chemotaxis, suggesting that cell populations exhibiting these dynamics can engage in distributed learning.
Furthermore, we derive two qualitatively different equations for optimal cue modulation, which optimally couples with linear and logarithmic sensing, respectively.
Finally, we compare a single smart agent, which possesses the internal memory of entire environment, with a population of non-smart agents, which has no internal memory yet can learn collectively via an external cue in the environment.
demonstrating that the population can achieve accumulation at the target more robustly. 
Our results provide a quantitative foundation for understanding the collective information processing mediated by dynamic sensing and modulation of cues.

\section{Model}
\label{sec:model-intro}

\subsection{Formulation of the collective exploration problem by an agent population}
\label{sec:RL-formulation}
To formulate the environmental exploration by a population of biological agents including cells,
we represent the environment as a connected undirected graph $G = (\mathcal{S}, \mathcal{E})$ with a finite set $\mathcal{S}$ of vertices and a finite set $\mathcal{E}$ of edges, which allows to incorporate general spatial topologies explored by the agent population.
In the environment, $N\in\mathbb{Z}_+$ homogeneous and independent agents are located at the vertices of the graph with overlap.
The location of agent $i\in \{1,2,\ldots,N\}$ at time $t\in\{0,1,\ldots\}$ is indicated by state $s^i_t\in\mathcal{S}$.
Within a small unit time, agent $i$ stochastically chooses either to stay at $s'=s_t^i$ or to move to a neighboring vertex $s'\in\mathcal{N}(s_t^i)$ with probability $\pi_t^i(s'|s_t^i)\in[0,1]$, where $\mathcal{N}(s)$ is the set of neighboring vertices of $s$. 
The transition probability $\pi_t^i:\mathcal{S}\to\Delta(\mathcal{S})$ is referred to as the agent's policy.

Each agent is assumed to learn a better policy $\pi_t^i$ to reach and accumulate at targets in the environment through exploration.
If the agent is a higher organism with sufficient internal memory, it can memorize the past experience, update $\pi_t^i$, and guide itself individually.
However, if the agent is a simple organism such as a cell, a versatile memory is not necessarily available. 

Instead, the policy (stochastic behaviors) of a simple agent is typically controlled and guided by various external cues~\cite{Kolaczkowska2013NatRevImm}.
One of these cues is a primary exogenous guidance cue, which conveys information about the targets' positions and magnitudes.
We specifically refer to the cue generated by the targets as exogenous to discriminate it from the cue generated by agents.
Bacterial chemotaxis to sugars is an example in which an agent is guided by exogenous cues~\cite{Berg2004}.
We consider a single exogenous cue for simplicity, whose spatial distribution is represented by $r:\mathcal{S}\to\mathbb{R}_{\geq 0}$.
For a biological instance, $r(s)=\exp(-l(s;s^0))$ represents fMLP concentration originating from damaged tissue at $s^0$, at distance $l(s;s^0)$ from the tissue~\cite{Szatmary2017MolBiolCell}.
The concentration can also vary depending on the size of the damage~\cite{Uderhardt2019Cell}.
While exogenous cues may be affected by agents, we assume for simplicity that they are not altered by the agents and dependent only on the targets' positions and magnitudes.

While the exogenous cue is the primary information about the targets, they may decay rapidly away from the targets due to harsh environments and their instabilities, being less likely to spread over long distances.
To amplify the information of the decaying exogenous cue, cell populations often utilize secondary endogenous cues, which are generated and modulated by individual cells.
For example, neutrophils sense exogenous cues, such as fMLP, from damaged tissue, and they also actively secrete more diffusive endogenous cues, such as LTB4, packaged in vesicles or exosomes to convey the information to distant immune cells~\cite{Ng2011JID,Afonso2012DevCell,Glaser2021CurrOpiCellBiol}.
For simplicity, we here consider only a single endogenous cue, whose spatial distribution is represented by $Z_t:\mathcal{S}\to\mathbb{R}_{\geq 0}$.

Finally, we assume that the policy $\pi_t^i$ of agent $i$ at vertex $s_{t}^{i}$ and time $t$ is solely determined by the concentrations of the exogenous cue $r(s^i_t)$ at $s_{t}^{i}$ and the endogenous cue around its vicinity $(Z_t(s))_{s\in \{s^i_t\}^\cup \mathcal{N}(s^i_t)}$.
This corresponds to the situation that each agent has no internal memory and no direct interaction with the others. Yet, it can guide itself using locally sensed endogenous and exogenous cues, which effectively work as a share memory deployed in the environment.

\subsection{Characterization of optimal exploration}
While each agent lacks an internal memory, the agents may still achieve target exploration if they collectively and optimally modulate and sense the endogenous cue by utilizing the locally sensed exogenous cue.
To derive optimal degradation, production, and sensing of the endogenous cue for collective exploration,
we associate the concentration of the exogenous cue $r$ with the reward function in RL framework because reaching the targets is equivalent to reaching the place where the exogenous cue $r$ is maximum.
Then, we formulate the problem as learning the optimal policy $\pi^*=\argmax_\pi \mathcal{J}[\pi]$ that maximizes the following expected cumulative reward $\mathcal{J}[\pi]$ with discounting and control penalty:
\begin{align}
    \label{eq: cost-EMDPs}
    \mathcal{J}[\pi]&\coloneqq E\left[\sum_{t=0}^\infty\gamma^tR^\pi(s_t,s_{t+1})\right],\\
    R^\pi(s_t,s_{t+1})&=r(s_t)-\frac{1}{\beta}\log\frac{\pi(s_{t+1}|s_{t})}{p(s_{t+1}|s_t)},
    \label{eq:temp-rewpen}
\end{align}
where the expected value  $E[f(s_{0:t})]$ is the path average of $f(s_{0:t})$ on the state sequences $s_0,s_1,\ldots s_t$. The discount factor $\gamma \in [0,1)$ is an intrinsic parameter common to all agents, weighting the current reward more than the future one.
It would be natural for simple organisms with limited lifespan to prioritize the present over the distant future.
The additional term $\frac{1}{\beta}\log\frac{\pi}{p}$ is a motility control penalty for the cue-dependent policy $\pi$ to deviate from the intrinsic (stationary) transition probability $p:\mathcal{S}\times \mathcal{S}\to [0,1]$.
The intrinsic transition probability would correspond to the random motility spontaneously exhibited by cells without cues. 
This form of cost was employed as biological control costs in bacterial chemotaxis~\cite{Nakamura2022PRR} and immunological learning~\cite{KatoPRR2021}. 
Note that the inverse control weight $\beta\in(0,\infty)$ represents the ease of motility control.
The formulation in Eq.~\eqref{eq: cost-EMDPs} is equivalent to the entropy-regularized Markov decision process $(\mathcal{S},r,p,\gamma,\beta)$~\cite{Todorov2009PNAS} (Appendix.~\ref{sec:entropy-MDPs}).
Defining the state value function $V^\pi:\mathcal{S}\to\mathbb{R}$ as
\begin{align}
    \label{eq:def-value}
    V^\pi(s)&\coloneqq E\left[\sum_{t=0}^\infty\gamma^t R^\pi(s_t,s_{t+1})\relmiddle|s_0=s\right],\quad s\in\mathcal{S},
\end{align}
and the optimal state value function $V^*:\mathcal{S}\to\mathbb{R}$ as $V^*(s)\coloneqq \max_\pi V^\pi(s)$ for $s\in\mathcal{S}$,
the optimal policy $\pi^*$ can be explicitly expressed using $V^*$ as follows:
\begin{equation}
    \pi^*(s'|s)=\frac{p(s'|s)\exp(\beta \gamma V^*(s'))}{\sum_{s'\in\mathcal{S}}p(s'|s)\exp(\beta \gamma V^*(s'))},\quad s,s'\in\mathcal{S}.
    \label{eq:EMDP-greedy-policy}
\end{equation}
The functional form of optimal policy means that the agent should bias its intrinsic transition $p(s'|s)$ to the surrounding vertices with higher $V^*(s')$. 
This behavior can be interpreted as a biased random walk along the local gradient of $V^*(s')$. 
Note that $V^*$ depends on the exogenous cue and environment geometry, both of which are initially unknown to the agents.
Consequently, the agents have to learn this optimal $V^*$ by updating a provisional, non-optimal $V_t$. 
If agents encode $V_t$ through the endogenous cue concentration $Z_t$ such that $V_t(s)=\mathcal{F}[Z_t](s)$, they can update $V_t(s)$ by modulating the endogenous cue via its degradation and production whilst moving by following the suboptimal policy with $V_t$ in Eq.~\eqref{eq:EMDP-greedy-policy}.
When $\mathcal{F}[Z_t(s)]$ closely approximates  $V^*(s)$, agents  move to vertices with higher state values, thus approaching the target.

\subsection{Optimal coupling of cue sensing and modulation}
\label{sec:distZ}
The optimal cue modulation is dependent on the actual form of the mapping $\mathcal{F}$, which boils down to the optimal coupling of chemotatic sensing behavior and cue modulation as we demonstrate in this subsection.

A biologically relevant mapping is the one associating the logarithm of the endogenous cue concentration with the estimate of the optimal state value: $V_t(s)=\mathcal{F}[Z_t](s)\coloneqq\log(Z_t(s))/\beta$ for $s\in\mathcal{S}$. 
This mapping aligns with Weber–Fechner response observed across a wide range of concentrations in various chemotactic cells~\cite{Janetopoulos2004PNAS, Kalinin2009BioPhys,Lazova2011PNAS,Tu2018AnnRevCondens}.
Under this logarithmic mapping, the corresponding agents' policy is reduced to \begin{equation}
    s^i_{t+1}\sim \pi_{t}(\cdot|s^i_{t})\coloneqq \frac{p(\cdot|s^i_{t})Z^{\gamma}_{t}(\cdot)}{\sum_{s'\in\mathcal{S}}p(s'|s^i_{t})Z^{\gamma}_{t}(s')},
    \label{eq:dist-Z-learning-policy}
\end{equation}
Meanwhile, the endogenous cue $Z_t$ should be modulated by its degradation and production for learning $Z_{t}$ which better approximates $V^*$.
We consider the simplest and biologically plausible dynamics of $Z_{t}$ that approximately minimize the deviation between $V_t=\mathcal{F}[Z_t]$ and $V^*$
under the physical constraint of diffusivity of $Z_t$ (see Appendix.~\ref{sec:algorithm-derivation} for derivation):
\begin{subequations}\label{eq:dist-Z-learning-diff}
    \begin{alignat}{3}
        Z_{t+1}(s)&-Z_{t}(s)=&-&\sum_{i=1}^N\mathbbm{1}_{s^i_{t}=s}\alpha\Delta Z_{t}(s^i_t)
        \label{eq:dist-Z-diff-learn}\\
        &&-&D \sum_{s'\in \mathcal{S}}L_{ss'} Z_{t}(s'),\quad s\in\mathcal{S}
        \label{eq:dist-Z-diff-regular}\\
        \Delta Z_t(s)&\coloneqq Z_{t}(s)&-&\exp(\beta r(s))\sum_{s'\in\mathcal{S}}p(s'|s)Z^\gamma_{t}(s'),
        \label{eq:dist-Z-diff-TD}
    \end{alignat}    
\end{subequations}
where $\mathbbm{1}_{s^i_{t}=s}$ is an indicator function that takes $1$ if agent $i$ is at vertex $s$ at time $t$, and $0$ otherwise; $\alpha\in(0,1/N)$ is the speed of the endogenous cue degradation and production by agents; and $L$ is the graph Laplacian of the environment $\mathcal{G}$.
In Eq.~\eqref{eq:dist-Z-diff-learn}, each agent at $s^i_t$ individually modulates $Z_{t}(s^i_t)$ by  $\alpha \Delta Z_t(s^i_t)$, which is dependent exponentially on the locally sensed $r$ and nonlinearly on surrounding $Z_t$.
In addition, $Z_t$ diffuses to neighbors independently of the agents at a constant rate $D$
(Figure~\ref{fig:landscape} (b)).

This learning dynamics is a distributed~\cite{Mnih2016PMLR} and sequentially regularized version of greedy Z-learning~\cite{Todorov2009PNAS}, which have two advantages from both biological and learning theory perspectives:
First, the optimal behavior learned by the dynamics accommodates the cost of agent's motility control as KL divergence from the intrinsic motility (Eq.~\eqref{eq:temp-rewpen}).
Second, it allows learning of the optimal policy~\eqref{eq:EMDP-greedy-policy} through a sequential process whilst maintaining an exploratory policy~\eqref{eq:dist-Z-learning-policy} like the Boltzmann policy. 
We refer to this coupling of cue sensing~\eqref{eq:dist-Z-learning-policy} and modulation~\eqref{eq:dist-Z-learning-diff} as the log-exp coupling model, from the biological interpretation of its continuous limit~\eqref{eq:log-continuous} obtained in the next subsection.

\subsection{Continuous limit of optimal cue sensing and modulation}
To demonstrate the biological relevance of the policy~\eqref{eq:dist-Z-learning-policy} and cue dynamics~\eqref{eq:dist-Z-learning-diff}, we derive its continuous limit, which is structurally equivalent to the well-established Keller-Segel models for collective chemotaxis~\cite{KELLER1971225,Tindall2008BullMathBiol,Painter2019JTB}.

For simplicity, we consider a one-dimensional lattice $\mathcal{S}=S^h\coloneqq \{0,\pm h,\pm 2h,\ldots\}$ with a uniform spacing $h$ and define the agents' intrinsic transition probability $p^h:S^h\times S^h\to[0,1]$ as a lazy random walk, such that $p^h(x|x)=\varepsilon^h(x)=O(h),\ p^h(x\pm h|x)=(1-\varepsilon^h(x))/2$.
Under an appropriate time scale and in the limit of infinite agents, the continuous space-time limits $h\to 0$ of the log-exp coupling models becomes the following pair of one-dimensional partial differential equations (derivation is provided in Appendix~\ref{sec:continuous-limit}):
\begin{subequations}               
    \label{eq:log-continuous}
    \begin{align}
        \pt_t\mu&=-\nabla\cdot\left(\mu\gamma \sigma^2 \nabla \log Z\right)+\frac{\sigma^2}{2}\nabla^2\mu,\\
        \pt_t Z&=-\left[\alpha \mu Z\right]+\left[\alpha\mu \exp(\beta r) Z^\gamma\right]+D\sigma^2\nabla^2 Z,
    \end{align}
\end{subequations}
where $\mu(t,x)\in\mathbb{R}_{\geq 0}$ is the agents density, and $Z(t,x)\in\mathbb{R}_{\geq 0}$ and $r(x)\in\mathbb{R}_{\geq 0}$ are the endogenous and exogenous cue concentrations, respectively. 
The parameter $\sigma \in \mathbb{R}_{\geq 0}$ is the randomness of the agents' motility, determined by the scaling of time and space. 

Equation~\eqref{eq:log-continuous} is structurally homologous to the Keller-Segel models~\cite{KELLER1971225,Tindall2008BullMathBiol,Painter2019JTB}, which was originally introduced to model microbial populations that migrate by senging and degrading an attractant cue:
\begin{subequations}
    \label{eq:KS}
    \begin{align}
        \pt_t n&=-\nabla \cdot \left(\chi n\nabla \log s\right)+D_n\nabla ^2n,\label{eq:KS-popu}\\
        \pt_t s&=-\left[\eta ns\right] + \left[-\kappa s\right]+D_s\nabla^2 s\label{eq:KS-chem},
    \end{align}
\end{subequations}
where $n(t,x)$ and $s(t,x)$ denote the microbial density and the attractant concentration, respectively. 
The parameter $D_n \in \mathbb{R}_{\geq 0}$ represents the randomness of the microbes, and $D_s \in \mathbb{R}_{\geq 0}$ is the diffusion coefficient of the attractant.
Apart from linear diffusion terms appearing in both equations, Eqs.~\eqref{eq:log-continuous} and \eqref{eq:KS} share the logarithmic sensing terms, $\nabla \log Z$ and $\nabla \log s$ and the density-dependent linear degradation term, $-\left[\alpha \mu Z\right]$ and $-\left[\eta ns\right]$.
The original Keller-Segel model does not include the production term, and its modified models introduced the term heuristically without considering the correspondence with the sensing term~\cite{Newman2004PRE,Pezzotta2018PRE}. 
In contrast, the production term we derived, $\left[\alpha\mu \exp(\beta r) Z^\gamma\right]$, has an explicit meaning as optimally coupled cue production with logarithmic sensing.
Notably, our production term depends on the endogenous cue concentration itself, which aligns with the signal relay mechanism in neutrophil swarming~\cite{Strickland2024DevCell}

\begin{table*}[tbp]
    \label{table:cont-comparison}
    \centering
    \begin{tabular}{|c c||c|c|c|}
        \hline & & Keller-Segel model& 
        log-exp coupling
        &
        lin-lin coupling
        \\
        \hline\hline 
        \begin{tabular}{c}
            Mapping from \\
            concentration to value 
        \end{tabular}
        & $\mathcal{F}[Z]$ & - & $\log(Z)/\beta $ & $Z/\beta $\\
        Chemotactic response & $\psi(\nabla Z,\nabla \log Z)$ & $\chi \nabla \log Z$ & $\gamma \sigma^2 \nabla \log Z$ & $\gamma \sigma^2 \nabla  Z$ \\
        Degradation & $H(\mu,Z)$ & $ (\eta\mu+\kappa)Z$ & $\alpha \mu Z$ & $ \alpha (1-\gamma)\mu Z$ \\
        Production & $G(\mu,Z,r)$ & (in extensions, $\mu r$) & $ \alpha\mu \exp(\beta r) Z^\gamma$ & $ \alpha \beta \mu r$ \\
        \hline
    \end{tabular}
    \caption{Component-wise comparison via \eqref{eq:propose-model-cont} of the continuous limits of the two learning dynamics with the classical Keller-Segel equation~\eqref{eq:KS}. The mapping $\mathcal{F}[Z]$ associates the endogenous cue concentration with the estimated state value. The term $\psi(\nabla Z,\nabla\log Z)$ represents the chemotactic response to endogenous cue concentration, and $H(\mu,Z)$ and $G(\mu,Z,r)$ do the degradation and production of the endogenous cue, respectively.
    Phenomelogical meanings of parameters are chemotactic coefficient $\chi$, cellular degradation rate $\eta$, and decay rate $\kappa$. 
    Computational meanings of parameters are inverse control weight $\beta$, discount factor $\gamma$, agent's motility $\sigma$ and learning rate $\alpha$.
    }
\end{table*}

The optimality of Keller-Segel dynamics was highlighted in Pezzotta \textit{et al.}~\cite{Pezzotta2018PRE} where, by using optimal control theory, they demonstrated that a cell population ascending an optimized endogenous cue gradient can be interpreted as an optimally controlled agent population following the optimal state value gradient. Yet, the production of the cue in their model was not theoretically principled. 
Our result extends their approach to reinforcement learning theory by considering an agent population where each agent ascends the endogenous cue gradient whilst cooperatively updating the endogenous cue (the estimated state value) based on the sensed exogenous cue.
This extension encompasses not only the optimal behaviors but also the learning process toward optimality, enabling the derivation of the optimal relationship between sensing and modulation of the endogenous cue gradient.

\subsection{Another optimal coupling}
\label{sec:distG}
The log-exp coupling is obtained by assuming the logarithmic mapping of $Z_{t}$ to $V_{t}$. 
We can obtain another coupling by considering a different mapping, which could expand the applicability of our theory to diverse biological situations. 

Another biologically relevant mapping is the one that linearly associates the endogenous cue concentration with the estimate of the state value $\mathcal{F}[Z_t](s)\coloneqq Z_t(s)/\beta$.
This mapping aligns with the absolute concentration sensing observed in various chemotactic behaviors~\cite{Kalinin2009BioPhys,Lazova2011PNAS,Fuller2010PNAS,Perelson2018}.
The corresponding policy becomes  
\begin{equation}
    s^i_{t+1}\sim \pi_{t}(\cdot|s^i_{t})\coloneqq \frac{p(\cdot|s^i_{t})\exp(\gamma Z_t(\cdot))}{\sum_{s'\in\mathcal{S}}p(s'|s^i_{t})\exp(\gamma Z_t(s'))}.
    \label{eq:dist-G-learning-policy}
\end{equation}
Similarly to the log-exp coupling, the optimal cue modulation dynamics is obtained (see Appendix.~\ref{sec:algorithm-derivation} for derivation):
\begin{subequations}\label{eq:dist-G-learning-diff}
    \begin{alignat}{3}
        Z_{t+1}(s)&-Z_{t}(s)=&-&\sum_{i=1}^N\mathbbm{1}_{s^i_{t}=s}\alpha\Delta Z_{t}(s^i_t)
        \notag\\
        &&-&D \sum_{s'\in \mathcal{S}}L_{ss'} Z_{t}(s'),\quad s\in\mathcal{S},\\
        \Delta Z_t(s)&\coloneqq Z_{t}(s)&-&\beta r(s)-\ln\sum_{s'\in\mathcal{S}}p(s'|s)\exp(\gamma Z_{t}(s')).
        \label{eq:dist-G-diff-TD}
    \end{alignat}
\end{subequations}
This learning dynamics is equivalent to a distributed version of greedy G-learning~\cite{Fox2016} with a sequential regularization.

The continuous limit of Eqs.~\eqref{eq:dist-G-learning-policy} and \eqref{eq:dist-G-learning-diff} is obtained (See Appendix~\ref{sec:continuous-limit} for derivation) as 
\begin{subequations}    \label{eq:lin-continuous}
    \begin{align}
        \pt_t\mu&=-\nabla\cdot\left(\mu\gamma \sigma^2 \nabla  Z\right)+\frac{\sigma^2}{2}\nabla^2\mu,\\
        \pt_t Z&=-\left[\alpha (1-\gamma)\mu Z\right]+\left[\alpha \beta \mu r \right]+D\sigma^2\nabla^2 Z.
    \end{align}
\end{subequations}
 Since the cue sensing and production are linear with respect to $Z$ and $r$, respectively, we refer to this coupling of sensing ~\eqref{eq:dist-G-learning-policy} and modulation~\eqref{eq:dist-G-learning-diff} as the linear-linear (lin-lin) coupling.
 Although production term is not included in the classical Keller-Segel model~\eqref{eq:KS-chem}, the linear production term $\alpha\beta\mu r$ of the lin-lin coupling model~\eqref{eq:lin-continuous} corresponds to the production term $\mu r$ of heuristically extended models~\cite{Newman2004PRE,Pezzotta2018PRE}.

By introducing chemotactic response function $\psi(\nabla Z, \nabla \log Z)$, cue-degradation rate function $H(\mu, Z)$, and cue-production function $G(\mu, Z, r)$, the log-exp couling, lin-lin coupling, and Keller-Segel models are integratively described as 
\begin{subequations}    \label{eq:propose-model-cont}
    \begin{align}
        \pt_t\mu&=-\nabla\cdot\left(\mu\psi(\nabla Z,\nabla \log Z)\right)+\frac{\sigma^2}{2}\nabla^2\mu,
        \label{eq:propose-model-cont-cell}\\
        \pt_t Z&=-H(\mu,Z)+G(\mu,Z,r)+D\sigma^2\nabla^2 Z,
        \label{eq:propose-model-cont-attractant}
    \end{align}
\end{subequations}
where the functional forms of $\psi(\nabla Z, \nabla \log Z)$, $H(\mu, Z)$, and $G(\mu, Z, r)$ are summarized in TABLE I.
From this table, we can identify the correspondences of two couplings with phenomenological Keller-Segel models.

Then, the next question would be how the optimal couplings affect the efficiency of the collective exploration behaviors by agents. 
Through numerical experiments in the following sections~\ref{sec:demostration} and \ref{sec:optimality}, we demonstrate the importance of the coupling between cue sensing and modulation in achieving optimal migration.

\section{Numerical experiments}
\label{sec:numexp}
In this section, we conduct numerical experiments using the derived learning dynamics to evaluate their behaviors and optimality in biologically relevant problems.

\subsection{Maze-solving problem and impact of optimal coupling}
\label{sec:demostration}
Biological systems often encounter exploration problems in complex environments. 
For instance, neutrophils detect and accumulate at infected tissues~\cite{Kolaczkowska2013NatRevImm}, and germ cells migrate to the gonad in embryos ~\cite{Renault2010Dev}.
Navigation in a complex environment can be framed as maze-solving, and various cells were experimentally verified to navigate themselves in artificial mazes~\cite{Scherber2012IntegrBiol,Tweedy2020Science, Phan2020PRX}.
In this section, we employ a complex maze for investigating the properties and optimality of the derived learning dynamics.

We adopted a maze with $|\mathcal{S}|=274$ vertices shown in Fig.~\ref{fig:distZ-diff_maze_onetar_ZZ-ZV}, which models the maze used in an experimental work for maze solving by \textit{Dictyostelium}~\cite[Figure 6]{Tweedy2020Science}.
In the following numerical experiments, we assume a lazy random walk for the intrinsic dynamics $p$ of agents: $p(s|s) = \varepsilon$ and $p(s'|s) = (1-\varepsilon)/|\mathcal{N}(s)|$, where $s' \in \mathcal{N}(s)$.
The parameter $\varepsilon \in [0,1)$ represents the sluggishness of migration, reflecting that the migration speed varies across species, even in the same geometry~\cite{Tweedy2020Science}.
The exogenous cue concentration $r$ is assumed to be high only at a goal (target) vertex $s^\mathrm{goal}\in\mathcal{S}$ and uniformly low at the other vertices
: $r(s^\mathrm{goal})=r_\mathrm{target}$ and $r(s)=r_\mathrm{default}$ for $s\neq s^\mathrm{goal}$.
This represents that the exogenous cue is produced only at $s^\mathrm{goal}$ and decays rapidly away from there.
Initially, the endogenous cue is uniformly distributed ($Z_0(s)= Z_0$ for $s\in\mathcal{S}$), and all $N$ agents start exploring from the vertex $s^i_0=s^\mathrm{start}\in\mathcal{S}$.
Then, agents collectively modify the endogenous cue concentration to reach the goal $s^\mathrm{goal}$ (Figure~\ref{fig:distZ-diff_maze_onetar_ZZ-ZV}).

\begin{figure*}[tbp]
    \hfill
    \begin{minipage}[b]{\linewidth}
        \centering
        \includegraphics[width=\textwidth]{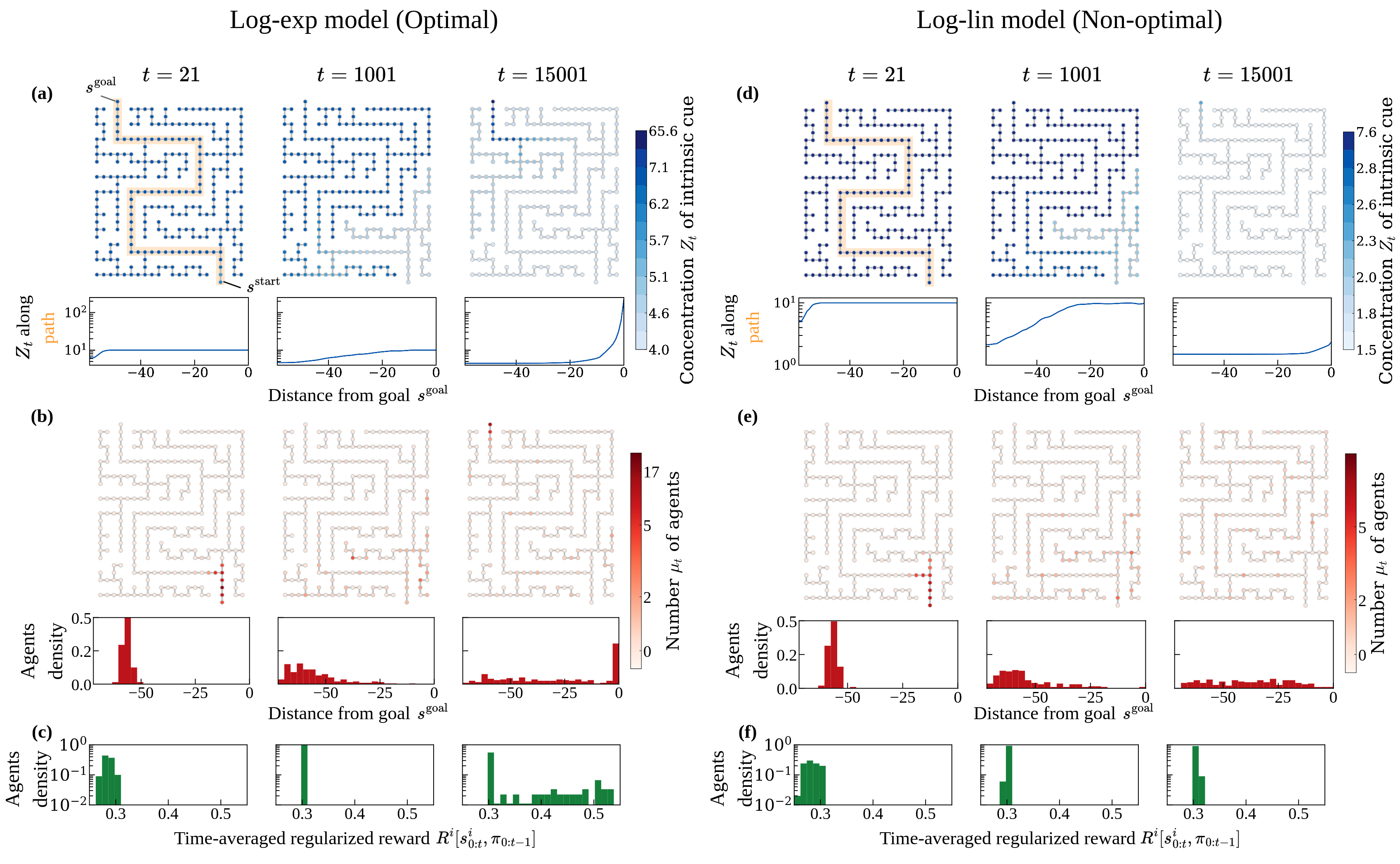}
    \end{minipage}
    \hfill
    \caption{
    Snapshots at $t=21,\, 1001,\, 15001$ for one trial of the maze exploration by the optimal log-exp coupling dynamics~\eqref{eq:dist-Z-learning-policy}, \eqref{eq:dist-Z-learning-diff} and the non-optimal log-lin coupling model~\eqref{eq:dist-Z-learning-policy}, \eqref{eq:dist-G-learning-diff}. 
    The distributions of the endogenous cue $Z_t$ are shown on the maze (upper panel) and on the path from $s^\mathrm{start}$ to $s^\mathrm{goal}$ (lower panel) for the log-exp (a) and the log-lin (d) couplings. 
    The path is highlighted with orange background in the maze and the distribution along the path is shown as a function from the goal $s^\mathrm{goal}$.
    The agent distribution $\mu_t:\mathcal{S}\to\{0,1,\ldots,N\}$ defined as $\mu_t(s)\coloneqq\sum_{i=1}^N\mathbbm{1}_{s_t^i=s}$ is similarly shown for the log-exp (b) and the log-lin (e) couplings. 
    Time-averaged regularized reward distribution $(R^i[s^i_{0:t},\pi_{0:t-1}])^N_{i=1}$ for the log-exp (c) and the log-lin (f) couplings.
    The vertices are color-coded by the value of cue concentration in (a) and (c) and that of agent density in (b) and (e).
    The parameter values used are $\alpha=0.0098$, $N=100$, $D=0.01$, $\gamma=0.8$, $Z_0=10.0$, $\varepsilon=0.5$, $\beta$, $r_\mathrm{target}=1.5$, $r_\mathrm{default}=0.3$.
    }
    \label{fig:distZ-diff_maze_onetar_ZZ-ZV}
\end{figure*}

Figure~\ref{fig:distZ-diff_maze_onetar_ZZ-ZV} (a), (b), and (c) show temporal snapshots of the concentration distribution $Z_t$ (a), agent distribution $\mu_t$ (b), and time-averaged regularized reward distribution $(R^i[s^i_{0:t},\pi_{0:t-1}])^N_{i=1}$ (c) for one trial of the log-exp coupling dynamics~\eqref{eq:dist-Z-learning-policy} and \eqref{eq:dist-Z-learning-diff}.
The time-averaged regularized reward for agent $i$ at time $t$ is defined as
\begin{equation}
    \label{eq:cum-reg-reward}
    R^i[s^i_{0:T},\pi_{0:T-1}]\coloneqq \frac{1}{T}\sum_{t'=0}^{T-1}\left(r(s^i_{t'})-\frac{1}{\beta}\log\frac{\pi_{t'}(s^i_{t'+1}|s^i_{t'})}{p(s^i_{t'+1}|s^i_{t'})}\right),
\end{equation}
which measures the difference between the cumulative exogenous cue concentration sensed by the agent and the motility control cost incurred by the agent from the start to time $T$ along the individual state history.
While the time-averaged reward $R^i[s^i_{0:T},\pi_{0:T-1}]$ assigns less weight on early rewards than the original discounted cumulative reward $\mathcal{J}[\pi]$~\eqref{eq: cost-EMDPs}, the maximizer of $R^i[s^i_{0:T},\pi_{0:T-1}]$ can still correspond to the maximizer of $\mathcal{J}[\pi]$ under certain conditions (c.f.~\cite[Sec. 8.4]{PutermanMDP1994}).

We observe that the agents expand their exploration area by dispersing at branching vertices and shaping a gradient toward unexplored directions via degradation of the endogenous cue  (Fig.~\ref{fig:distZ-diff_maze_onetar_ZZ-ZV} (a) and (b) at $t=21$ and $t=1001$).
Upon reaching $s^\mathrm{goal}$, they produce the endogenous cue, reinforcing the gradient toward $s^\mathrm{goal}$.
Eventually, the gradient aligns from $s^\mathrm{start}$ to $s^\mathrm{goal}$ and major fraction of agents accumulate near $s^\mathrm{goal}$ (Fig.~\ref{fig:distZ-diff_maze_onetar_ZZ-ZV} (a) and (b) at $t=15001$).
As a result, even in a complex environment, the agents can accumulate at $s^\mathrm{goal}$ through the endogenous cue degradation and production, resulting in higher $R^i$ for many agents (Fig.~\ref{fig:distZ-diff_maze_onetar_ZZ-ZV} (b) and (c) at $t=15001$).
These behaviors are qualitatively consistent with experimentally observed phenomena, such as exploration via endogenous cue degradation \cite{Insall2022TrdCellBiol} and accumulation at target through endogenous cue production~\cite{Oliveira2016NatRevImm}.
Similar behaviors are also observed in the optimal lin-lin coupling dynamics~\eqref{eq:dist-G-learning-policy}, \eqref{eq:dist-G-learning-diff} (see Fig.~\ref{fig:distZ-diff_maze_onetar_VV-VZ}).

To evaluate the importance of the optimal coupling of the cue sensing and modulation, we investigated an agent population with a non-optimal coupling of the logarithmic sensing policy~\eqref{eq:dist-Z-learning-policy} and the linear cue production~\eqref{eq:dist-G-learning-policy}, i.e., the log-lin coupling.
Figure ~\ref{fig:distZ-diff_maze_onetar_ZZ-ZV} (d) and (e) show that the agents with the log-lin coupling disperse, degrade the endogenous cue, and generate a cue gradient toward $s^\mathrm{goal}$ up to around $t=1001$, similarly to the optimal log-exp coupling.
However, the generated endogenous cue gradient is shallow and most agents fail to accumulate at the target even after the cue concentration has almost reached the steady state, resulting in a low $R^i$ for most agents (Fig.~\ref{fig:distZ-diff_maze_onetar_ZZ-ZV} (d)-(f) at $t=15001$).
For another non-optimal coupling of linear sensing policy~\eqref{eq:dist-G-learning-policy} and exponential production~\eqref{eq:dist-Z-learning-diff}, we observed an overaccumulation (see Figure~\ref{fig:distZ-diff_maze_onetar_VV-VZ}).
These results highlight the importance of optimal coupling between cue sensing and modulation for achieving efficient exploration and navigation in a complex environment.



\subsection{Coupling optimality and its consequence}
\label{sec:optimality}
To scrutinize how the optimality of coupling affects the agent's behaviors, we compare the steady states in a simpler one-dimensional environment for the four couplings; Each coupling is obtained by combinatorially pairing sensing~\eqref{eq:dist-Z-learning-policy} and \eqref{eq:dist-G-learning-policy}, with modulations~\eqref{eq:dist-Z-learning-diff} and \eqref{eq:dist-G-learning-diff}.
The environment is a graph with $20$ vertices $\mathcal{S}=\{s^0,s^1,\ldots,s^{19}\}$ connected in a chain, where a target locates at the vertex $s^{19}$ (Fig.~\ref{fig:comparison-steady} (a)).
We also modulate the relative importance of the reward and the control cost by changing the parameter $\beta$.

\begin{figure*}[tbp]
    \hfill
    \begin{minipage}[b]{\linewidth}
        \centering
        \includegraphics[width=\textwidth]{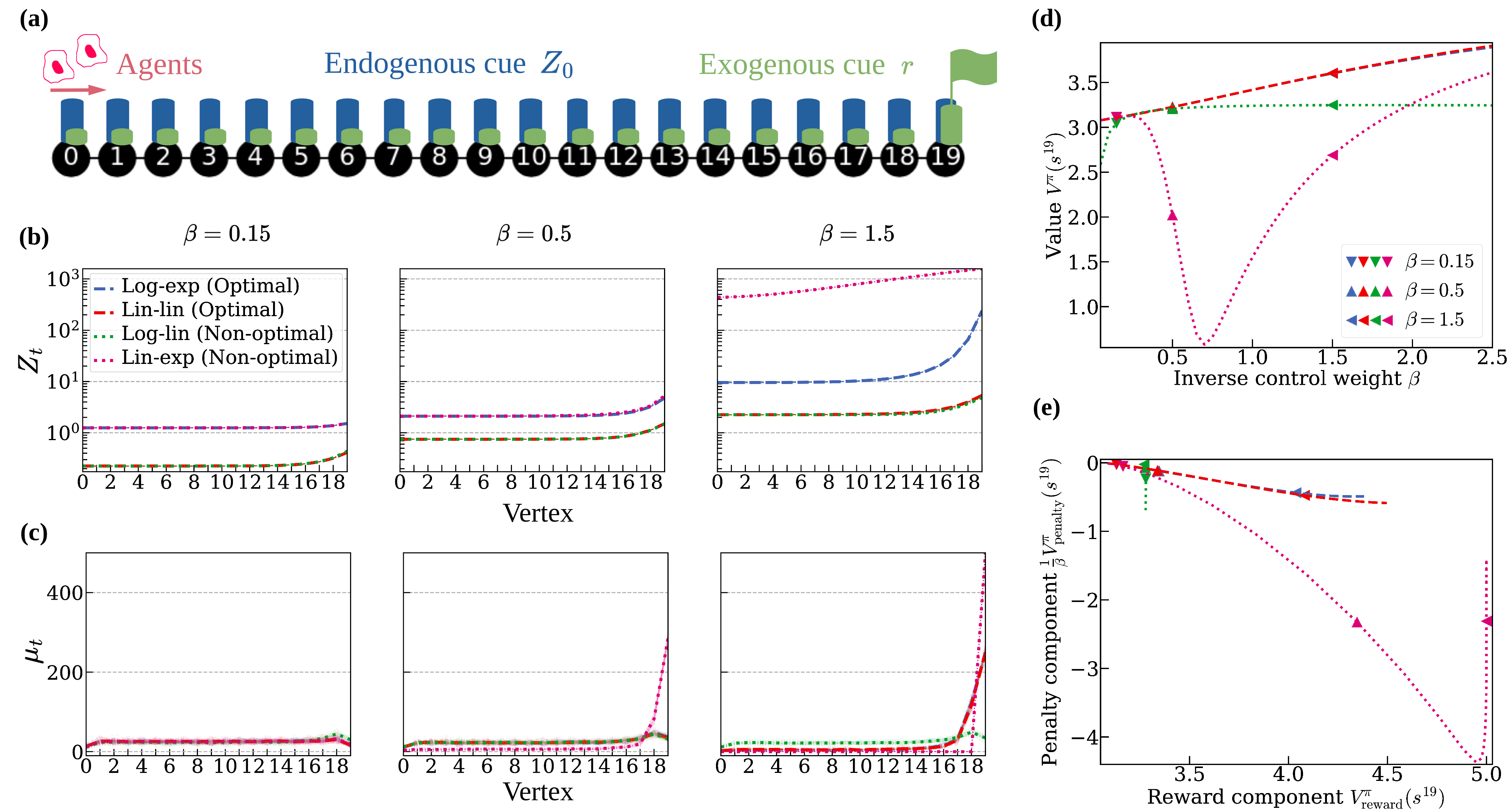}
    \end{minipage}
    \hfill
    \caption{
    (a) An environment represented by a graph with 20 vertices $\mathcal{S}=\{s^0,s^1,\ldots,s^{19}\}$ connected in a chain. 
    All agents start from vertex $s^{0}$, and the target is located only at vertex $s^{19}$. 
    The exogenous cue concentration (reward) is set $r(s^{19})=r_\mathrm{target}=1.0$ and $r(s)=r_\mathrm{default}=0.3$ for $s\neq s^{19}$. 
    (b,c) Endogenous cue concentration $Z_T$ (b) and agent distribution $\mu_T$ (c) at $T=8000$ for $50$ trials with $\beta\in\{0.15,0.5,1.5\}$ for the optimal and non-optimal couplings.
    The optimal couplings are log-exp coupling~\eqref{eq:dist-Z-learning-policy}, \eqref{eq:dist-Z-learning-diff} and lin-lin coupling~\eqref{eq:dist-G-learning-policy}, \eqref{eq:dist-G-learning-diff}; The non-optimal couplings are log-lin coupling~\eqref{eq:dist-Z-learning-policy}, \eqref{eq:dist-Z-learning-diff}, lin-exp coupling~\eqref{eq:dist-G-learning-policy}, \eqref{eq:dist-G-learning-diff}.
    Each trial's $Z_T$ and $\mu_T$ are shown with thin solid curves, and the trial averages $\tilde{Z}_T$ and $\tilde{\mu}_T$ are shown with thick solid curves.
    (d) State value $V^\pi(s^{19})$ at vertex $s^{19}$ for each policy $\pi[\tilde{Z}_T]$ based on the trial-averaged steady concentration $\tilde{Z}_T$ for the four couplings with $\beta\in[0.25,2.5]$. 
    (e) The relationship between the reward component $V^\pi_\mathrm{reward}(s^{19})$ and penalty component $\frac{1}{\beta}V^\pi_\mathrm{penalty}(s^{19})$ of $V^\pi(s^{19})$ for the four couplings with $\beta\in[0.25,2.5]$.
    The parameters values are $N=500,\varepsilon=0.5,\alpha=0.00196,D=0.01,\gamma=0.8,\beta\in[0.25,2.65]$.
    }
    \label{fig:comparison-steady}
\end{figure*}

For the optimal couplings, i.e., log-exp and lin-lin couplings, the steady agent distributions $\mu_T$ are almost identical while the steady cue distribution $Z_T$ are different, for each $\beta$ (Fig.~\ref{fig:comparison-steady} (c)). 
The agents accumulate more at the target as $\beta$ increases, indicating that accumulation is encouraged if the control cost becomes lower and the agents can deviate more from the intrinsic lazy random walk.

However, for the non-optimal couplings, i.e., log-lin and lin-exp couplings, dispersal and overaccumulation in $\mu_T$ were observed, respectively (Fig.~\ref{fig:comparison-steady} (b)), which are consistent with the results in the maze environment (Fig.~\ref{fig:distZ-diff_maze_onetar_ZZ-ZV} (e)).
For the log-lin coupling, the agents remain dispersed while the cue distribution is almost identical to that of the optimal lin-lin coupling. 
Because the logarithmic sensing requires a steeper gradient than linear one to achieve the same chemotactic behavior, the agents with log-lin coupling were less sensitive to the gradient than lin-lin coupling and thus failed to accumulate at the target.
For the lin-exp coupling, agents overaccumulated at the target compared with the optimal log-exp coupling when $\beta$ is high.
As linear sensing is more sensitive than logarithmic sensing especially when the basal cue concentration is high, this oversensitivity results in the overaccumulation. 
These results demonstrate that the disruption of the optimal coupling substantially alters the optimal exploration and exploitation behavior of agents.

To evaluate the optimality of the four couplings in the steady states more quantitatively, 
we compare the state value function
$V^\pi$~\eqref{eq:def-value} for each policy $\pi_T$, along with its reward component $V^\pi_\mathrm{reward}$ and the penalty (control cost) component $V^\pi_\mathrm{penalty}$.
Here we approximate the fluctuating steady concentration $Z_T$, which determines $\pi_T$, with the trial-average $\tilde{Z}_T$, and define 
$V^\pi_\mathrm{reward}(s)\coloneqq
E\left[\sum_{t=0}^\infty\gamma^t r(s_t)\relmiddle|s_0=s\right]$ ,
and $V^\pi_\mathrm{penalty}(s)\coloneqq-E\left[\sum_{t=0}^\infty\gamma^t\log\pi(s_{t+1}|s_{t})/p(s_{t+1}|s_{t})\relmiddle|s_0=s\right]$.
These values are obtained by value iteration method (Appendix.~\ref{sec:entropy-MDPs}), and additional state values are provided in Figure~\ref{fig:comparison-steady-supplement}.

The state value $V^\pi$, the optimization metric, is larger for the optimal log-exp and lin-lin couplings than for the non-optimal log-lin and lin-exp couplings at any $\beta$ (Figure~\ref{fig:comparison-steady} (d)). 
The reward component $V^\pi_\mathrm{reward}$ and the penalty component $V^\pi_\mathrm{penalty}/\beta$ of $V^\pi$ are balanced for the optimal couplings (Figure~\ref{fig:comparison-steady}(e)) such that $V^\pi_\mathrm{reward}$ and $V^\pi_\mathrm{penalty}/\beta$ trade off along the change in $\beta$.
Despite the different shapes of $Z_T$ (Figure~\ref{fig:comparison-steady}(b)), the two optimal couplings exhibit almost identical values of $V^\pi$ and achieve the same balancing curve of $V^\pi_\mathrm{reward}$ and $V^\pi_\mathrm{penalty}/\beta$ for all $\beta$.
This is a natural consequence of the fact that the two optimal couplings solve the same optimization problem, while their representations of $V$ via $Z$ are different.

In constrast, 
we observe low $V^\pi$ and imbalances between $V^\pi_\mathrm{reward}$ and $V^\pi_\mathrm{penalty}/\beta$ for the non-optimal couplings.
$V^\pi_\mathrm{reward}$ of the non-optimal log-lin coupling remains constant regardless of $\beta$, while the absolute value of $V^\pi_\mathrm{penalty}/\beta$ increases as $\beta$ decreases (Figure~\ref{fig:comparison-steady}(e)).
This indicates that agents dynamics are not navigated appropriately to exploit $V^\pi_\mathrm{reward}$ due to the non-optimal coupling.
For the non-optimal lin-exp coupling, $V^\pi_\mathrm{reward}$ is far below that of the optimal coupling and changes non-monotonously as a function of $\beta$ (Fig.~\ref{fig:comparison-steady}(d)).
Figure~\ref{fig:comparison-steady}(e) indicates that the lin-exp coupling incurs a greater penalty than the optimal ones to gain the same value of $V^\pi_\mathrm{reward}$ (Fig.~\ref{fig:comparison-steady}(e)).
This imbalance between $V^\pi_\mathrm{reward}$ and $V^\pi_\mathrm{penalty}/\beta$ results in a deficiency in overall performance.
This analysis verifies that non-optimal couplings are indeed inferior in terms of optimization for accumulating at targets, highlighting the importance to optimally couple cue sensing and modulation.

These results clarify that the optimal endogenous cue modulation qualitatively differs depending on the gradient sensing characteristics:
the logarithmic sensing policy optimally couples with production that is exponentially dependent on exogenous cue concentration, whereas the linear sensing policy optimally couples with production that is linearly dependent on the concentration.
These two are almost identical in terms of the optimization criterion, and there may exist other optimal couplings, corresponding to different choices of the mapping $\mathcal{F}[Z]$.
This flexibility in the mapping allows optimization-based modeling to incorporate biological details, thereby expanding its applicability. 

\subsection{Comparison with a single-agent system}
\label{sec:single-agent}
We have analyzed distributed learning in which agents with limited memory share and learn environmental information cooperatively through the endogenous cue to search for targets.
However, there exist higher organisms that have sufficient memory and computational capacity as a single agent to store and learn environmental information within each individual's internal model~\cite{Behrens2018Neuron}.
An unsolved issue in the evolution of biological information processing is the fitness advantages and disadvantages of distributed information processing by agents with limited computational capacity compared with the centralized one by a single smart agent with sufficient capacity~\cite{Vining2019RSocB,Sole2019RSocB,Navas2022TrendsCogSci}.
Even if both systems achieve the same objective, metrics such as robustness can depend on the system's architecture, and their comparative advantages can vary based on these metrics~\cite{Moses2019FrontImmuno} (c.f. No free lunch theorem for search~\cite{Wolpert1995SFI}).
We address this problem by leveraging the fact that both of our distributed model and the usual single-agent learning model are based on the same RL framework.

Specifically, we consider a single agent with an internal memory to store the estimate of state value function $Z_t:\mathcal{S}\to\mathbb{R}_{\geq 0}$ at each time step $t=0,1,\ldots$.
The agent has a state $s_t\in\mathcal{S}$ and transits between the states based on $Z_t$ following logarithmic sensing policy~\eqref{eq:dist-Z-learning-policy}, and update $Z_{t}$ as
\begin{equation}
    Z_{t+1}(s)-Z_{t}(s)=-\mathbbm{1}_{s_{t}=s}\alpha'\Delta Z_{t}(s),
    \label{eq:single-Z-learning-diff}
\end{equation}
which is nearly identical to the exponential cue production~\eqref{eq:dist-Z-learning-diff} with $N=1$, except the diffusion term~\eqref{eq:dist-Z-diff-regular}.
This update rule is also identical to the standard greedy Z-learning~\cite{Todorov2009PNAS}.
For comparison, the learning rate $\alpha'\in(0,1)$ for the single agent is set to $\alpha'=\alpha N$, indicating that it has the same learning capability as the entire agent population.

\begin{figure}[tbp]
    \hfill
    \begin{minipage}[b]{\linewidth}
        \centering
        \includegraphics[width=\textwidth]{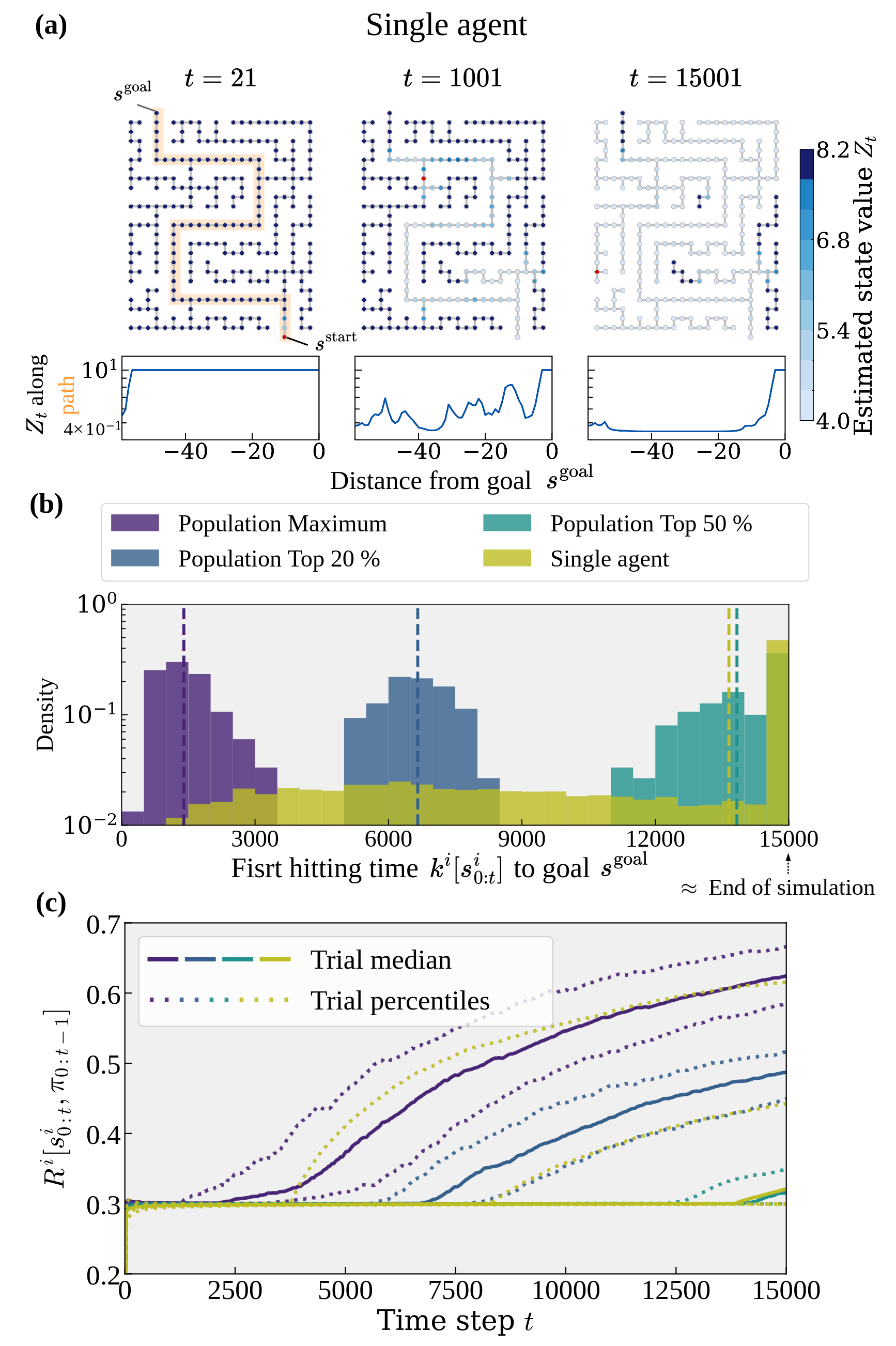}
    \end{minipage}
    \hfill
    \caption{
    (a) Snapshots of the estimated state value $Z_t$ at $t=21,1001,15001$ for one trial by a single smart agent with policy~\eqref{eq:dist-Z-learning-policy} and update rule~\eqref{eq:single-Z-learning-diff}. 
    The upper panels show $Z_{t}$ and the agent's location $s_t$ on the maze.
    The vertex $s_t$, where the agent is located, is marked in red, while the other vertices $s\neq s_t$ are color-coded based on $Z_{t}(s)$.
    The lower panels show the values of $Z_{t}$ along the path highlighted by orange as a function of the distance from the goal $s^\mathrm{goal}$.    
    (b) Distribution of the first hitting time to $s^\mathrm{goal}$ computed from multiple trials by the single agent (yellow), and by the agent at the maximum (purple), $20$th percentile (blue), and $50$th percentile (green) of the population.
    The first hitting time $k^i[s^i_{0:T}]$ is calculated as 
    $k^i[s^i_{0:T}]=\min\left\{t\in\{1,2,\ldots,T\}\relmiddle| s^i_t=s^\mathrm{goal}\right\}$ if $s^\mathrm{goal}\in\{s^i_t\}_{t=0:T}$; otherwise, $k^i[s^i_{0:T}]=T$.
    Each dashed line indicates the median value in the distributions. 
    (c) Temporal variation of the time-averaged regularized reward $R^i[s^i_{0:t},\pi_{0:t-1}]$ of the single agent (yellow) and of the agent at the maximum (purple), the $20$th percentile (blue), and $50$th percentile (green) of the population.
    Each solid line represents the trial median at each time step.
    Dashed lines represent the trial $10,30,70,90$th percentiles for the single agent (yellow), and the trial $10,90$th percentiles for each population percentiles (purple, blue, and green).
    The parameters used are $T=15001$, $(\alpha, N, D,\# \mathrm{trials})=(0.0098, 100, 0.01,150)$ for the agent population, and $(0.98, 1, 0,15000)$ for the single agent.
    The other parameters are same as Figure~\ref{fig:distZ-diff_maze_onetar_ZZ-ZV}.
    }
    \label{fig:distZ-diff_maze_onetar_vs-united_comparison}
\end{figure}

Figure~\ref{fig:distZ-diff_maze_onetar_vs-united_comparison}(a) shows the exploration behavior of the single agent, which should be compared with its distributed counterpart (Fig.~\ref{fig:distZ-diff_maze_onetar_ZZ-ZV} (a)). 
In this one-trial comparison, the overall update speed of $Z_t$ by the single agent is slower than the agent population due to its limited exploration range (Fig.~\ref{fig:distZ-diff_maze_onetar_vs-united_comparison}(a) $t=1001$).
At a time $t=15001$, while the agent population has mostly completed learning and some agents have accumulated at the target (Fig.~\ref{fig:distZ-diff_maze_onetar_ZZ-ZV} (a) and (b)), the single agent is still exploring the target.
Furthermore, the value distribution becomes irregular (Fig.~\ref{fig:distZ-diff_maze_onetar_vs-united_comparison}(a) $t=1001$) due to random exploration histories and the lack of a diffusion term.
These observations suggest that the single smart agent may not explore and learn as effectively as the less smart agent population.

For a more quantitative evaluation, Fig.~\ref{fig:distZ-diff_maze_onetar_vs-united_comparison} (b) summarizes the hitting time statistics for multiple trials. 
The first hitting time of the smart single agent has a broad distribution ranging over $[1500,15000]$, indicating that the success of the exploration by a single agent is extremely stochastic. 
For the agent population, the first hitting time fluctuates not only among trials, but also among agents in the population (Fig.~\ref{fig:distZ-diff_maze_onetar_vs-united_comparison} (b)).
However, if we focus on the agent who hit the target first among the population, 
its hitting time distribution for the multiple trials consolidates around $t=1500$ and its median is close to the shortest hitting time of the single agent among the trials. 
Although the median of this distribution increases as we focus on hitting times ranked at $20\%$ and $50\%$ among the population, the value of $50\%$ is still comparable to that of the single agent.
Thus, population-wise exploration can work more robustly than its single-agent counterpart.
We further evaluate the performance of the learning process using the time-averaged regularized reward $R^i$~\eqref{eq:cum-reg-reward} of each agent, which is related to the optimization metric $\mathcal{J}$~\eqref{eq: cost-EMDPs} (Figure~\ref{fig:distZ-diff_maze_onetar_vs-united_comparison} (c)).
The bottom 50\% of $R^i$ in the agent population is smaller than its median of the single agent's trials, which corresponds to the fraction of agents trapped in dead ends (Fig.~\ref{fig:distZ-diff_maze_onetar_ZZ-ZV} (b) $t=15001$).
However, the top 20\% of $R^i$ in the agent population are larger than the top 30\% of the single agent's trials for almost 90\% of trials.

These results indicate that population-wise exploration can work more robustly in each trial than the single-agent-wise counterpart at the cost of the fraction of agents who could not reach the target.
Such a fraction can be compensated effectively by the self-replication of the other fraction that reached the target given the cost for replication is cheap.
This suggests that distributed exploration becomes more advantageous for simple and fertile organisms.
In real biological problems such as neutrophil accumulation at a wound site, a prompt arrival to the target even by a fraction of immune cells is effective against exponentially growing pathogens entering from the wound.

\section{Summary and Discussion}
\label{sec:conclusion}
To formulate a cell population that distributedly explores and learns target information via chemical cues in the environment, we leveraged reinforcement learning theory and derived two qualitatively different optimal dynamics of the population with linear and logarithmic gradient sensing.
These dynamics have continuous limits that are structurally equivalent to the phenomenological Keller-Segel model, indicating that optimal learning could be achieved by biologically relevant cell populations. 
We also verified that the optimal coupling of the sensing and modulation of cues is crucial for achieving efficient exploration by the population.
Additionally, we have demonstrated that an agent population with limited computational capability can achieve more robust learning via distributed modulation of external cues than a smart single agent with complete internal memory.

\subsection{Extensions of the formulation}
Although we derived only two optimal couplings, our formulation can be extended beyond them by incorporating various biological factors, providing a foundation for the quantitative understanding of collective information processing mediated by extracellular signals.
For instance, we can derive a learning dynamics with a Hill function in chemotactic response and production, reflecting the underlying receptor kinetics~\cite{Tindall2008BullMathBiol,Cremer2019Nat,Bacon1965JMolBiol} (See Appendix.~\ref{sec:generalized-derivation}).
Furthermore, we may explore various collective dynamics from the principle of optimality by modifying (i) the cue-value mapping and control penalty, (ii) regularization term, and (iii) agent's interactions and heterogeneity.

Nonlinear mappings that align with control penalties enable the derivation of  nonlinear dynamics. 
In this work, we used a log-form control penalty (KL-divergence), which aligns with the logarithmic mapping, to obtain the log-exp model. 
Alternatively, other control penalties, such as Tsallis relative entropy~\cite{Chow2018ICML}, could be employed to explore different nonlinear dynamics through their associated nonlinear mappings.


The diffusivity of the endogenous cue, which was regarded as a physical constraint in Sec.~\ref{sec:algorithm-derivation}, can also be interpreted as Laplacian regularization, which ensures continuity and smoothness of the value between adjacent vertices in the graph~\cite{LangNIPS1989,Zhu2003ICML}.
We may exploit other types of regularization.
For instance, $L^2$ regularization ($\lambda \|Z_t\|^2_2$, $\lambda\in\mathbb{R}_{\geq 0}$) can suppress extreme values and results in another sequential regularization term $-\lambda Z_t(s)$. 
This aligns with the cue decay ($\lambda=\kappa$, Eq.~\eqref{eq:KS-chem}) in Keller-Segel models, or the pheromones evaporation in ant population dynamics~\cite{Chowdhury2005PhysLifRev}.
Moreover, agent-dependent regularization might be employed to account for the alternation of cue diffusivity via the enzyme produced by cells~\cite{Weickert1997Springer}.

Incorporation of agents' interactions and heterogeneity is more challenging theoretically.
We obtained homogeneous agents that interact only through one endogenous cue, which stems from our assumption that the reward function depends only on the targets' location and magnitude.
However, agents' behaviors can be influenced by integrated information from other agents, which is often biologically mediated via a variety of cues, including repellents~\cite{Hughes2018FEBS}.
For instances, immune responses are modulated by chemokines released by other immune cells to control inflammation~\cite{Hopke2020NatCom}. 
Microorganisms recognize and aggregate in response to the states of others to adapt to starvation~\cite{Crespi2001TreEcoEvol}.
In development, cells repel each other to form tissues~\cite{Stramer2017NatRevMolCellBiol}.
Moreover, biological systems often involve heterogeneous cell populations, where different subgroups play distinct roles, such as leaders navigating gradients at the front and followers trailing behind~\cite{theveneau2013NatCellBiol}.
Extending our formulation to include interactions~\cite{Lauriere2022ArXiv} and heterogeneity~\cite{Subramanian2020ArXiv} is essential to comprehend the computational aspects of complex and collective behaviors observed in heterogeneous biological populations.

\subsection{Evolutionary advantages of distributed computation}
Understanding the evolution of biological information processing 
lies at the intersection of biophysics, evolutionary biology, and computational neuroscience.
Our result demonstrates that a population of simple agents can excel at learning target information robustly compared to a single smart agent. 
This aligns with the redundancy principle~\cite{Schuss2019PhysLifeRev}, which states that having many particles with the same objective can increase the probability of quickly achieving the objective as a population.
This principle has been applied to various biological processes from synaptic signaling to embryo fertilization, where the arrival time of the fastest individuals is more critical than the population average.
According to this biophysical principle, 
distributed information processing may have an evolutionary benefit for agents with limited capabilities but low replication costs. 
For example, neutrophils may achieve a robust population response by dispersing numerous genetically identical cells.
Integrating learning and evolutionary theories~\cite{Turney1996EvolComp,Xue2016PNAS,Khadka2018NIPS,Nakashima2022PRR} would be a promising direction for advancing our understanding of the principles behind the emergence and persistence of diverse forms of biological information processing.
\begin{acknowledgments}
The authors thank Shuhei A. Horiguchi for discussions.
The first author received the JSPS Research Fellowship Grant No.24KJ0846 and was partially supported by the World-leading Innovative Graduate Study Program in Proactive Environmental Studies (WINGS-PES), The University of Tokyo. This work was supported by JST CREST JPMJCR2011, and JSPS KAKENHI Grant Numbers 24KJ0090, 24H02148.

\end{acknowledgments}

%

\onecolumngrid
\appendix 
\clearpage
\setcounter{section}{0}

\renewcommand\appendixname{Supplement}

\section{A brief review of reinforcement learning}
\setcounter{page}{1}
\setcounter{equation}{0}
\setcounter{figure}{0}
\renewcommand{\theequation}{S\arabic{equation}}
\renewcommand{\thefigure}{S\arabic{figure}}
\renewcommand{\thepage}{S\arabic{page}}

\label{sec:MDPs}
In this section, we review entropy-regularized Markov decision process (ERMDP) and entropy-regularized reinforcement learning (ERRL) used in the main text.
\subsection{Conventional MDP and RL}
This subsection reviews the infinite horizon discount Markov decision process (MDP) and conventional reinforcement learning (RL) based on ~\cite{PutermanMDP1994, BartoSutton2018}, which is the foundation of ERMDP and ERRL.

Conventional MDP is represented as an interaction process between an environment and an agent.
Initially, at time $t=0$, the environment samples an initial state $s_0$ from an initial distribution $p_0$.
At time $t=0,1,\ldots$, the agent observes the state $s_t$ from the environment and stochastically decides the action $a_t$ based on the behavior policy $\pi(\cdot|s_t)$ conditioned only on the current state $s_t$.
The environment then samples the next state $s_{t+1}$ according to the state transition probability $P(\cdot|s_t,a_t)$ conditioned on the state $s_t$ and the agent's action $a_t$, and returns a reward $r_t$ based on $s_t, a_t$ to evaluate how good the current state and action are.
In sum, MDP is mathematically defined as a tuple $(\mathcal{S},\mathcal{A},P,r)$ comprising the set of possible states $\mathcal{S}$ in the environment, the set of all possible actions $\mathcal{A}$ by the agent, transition probability $P:\mathcal{S}\times \mathcal{A}\to\Delta(\mathcal{S})$, and reward function $r:\mathcal{S}\times\mathcal{A}\to\mathbb{R}$.
If we assume $P(\cdot|s,a)$ is deterministic and $r:\mathcal{S}\to\mathbb{R}$, we can recover the formulation without explicit action employed, as adapted in the main text, where the agent directly controls state transition probability $P:\mathcal{S}\to\Delta(\mathcal{S})$ instead of taking explicit actions.
For simplicity, we consider $\mathcal{S}$ and $\mathcal{A}$ are finite and $r$ is bounded.

The main objective in MDP is to find a policy $\pi^\dagger$ that maximizes the cumulative sum of rewards.
For conventional MDP and RL, the objective is defined as expected cumulative discount reward $\mathcal{J}[\pi]$ with discount factor $\gamma\in[0,1)$:
\begin{equation}
    \mathcal{J}[\pi]\coloneqq E_{P^\pi[s_{1:\infty},a_{0:\infty}\mid s_0] p_0(s_0)}\left[\sum_{k=0}^\infty \gamma^kr (s_k,a_k)\right], \label{eq:cost-discMDP}
\end{equation}
where the expectation is defined via the probability measure $P^\pi[s_{1:\tau+1},a_{0:\tau}|s_0]$ over trajectories $(s_k,a_k)_{k=0,\ldots,\tau}$ under $\pi$ and $P$ and an initial state probability $p_0$:
\[P^\pi[s_{1:\tau+1},a_{0:\tau}|s_0]\coloneqq \prod_{k=0}^\tau P(s_{k+1}|s_k,a_k)\pi(a_k|s_k).\]
Note we use this notation for expectation to clarify the dependence on $\pi$ and $P$ in the rest of this section.
The objective~\eqref{eq:cost-discMDP} evaluates future rewards as smaller by $\gamma$, which can easily account for continuous tasks observed in RL problems.
It is reasonable for biological organisms with limited memory to evaluate the near future, and mammals have been believed to implement the criterion \eqref{eq:cost-discMDP} by using neurotransmitter concentrations as discount factor $\gamma$~\cite{Tanaka2004NatNeu}.

One representative method to find the optimal policy $\pi^\dagger\coloneqq \operatorname{argmax}_\pi \mathcal{J}[\pi]$ is the value-based approach, where indirectly determines $\pi^\dagger$ via state value functions or action value functions. 
The state value function $V^\pi:\mathcal{S}\to\mathbb{R}$ and action value function $Q^\pi:\mathcal{S}\times\mathcal{A}\to\mathbb{R}$ for a policy $\pi:\mathcal{S}\to\Delta(\mathcal{A})$ is defined as follows:
\begin{subequations}
    \label{eq:action-state-values}
    \begin{align}
        V^\pi(s)&\coloneqq E_{P^\pi[s_{1:\infty},a_{0:\infty}|s_0=s]}\left[\sum_{k=0}^\infty \gamma^kr (s_k,a_k)\right],\quad s\in\mathcal{S}\\
        Q^\pi(s,a)&\coloneqq E_{P^\pi[s_{2:\infty},a_{1:\infty}|s_1]P(s_1|s_0=s,a_0=a)}\left[\sum_{k=0}^\infty \gamma^kr (s_k,a_k)\right],\quad s\in\mathcal{S}, a\in\mathcal{A}.
    \end{align}
\end{subequations}
Define the optimal state value function $V^\dagger:\mathcal{S}\to\mathbb{R}$ and action value function $Q^\dagger:\mathcal{S}\times\mathcal{A}\to\mathbb{R}$ as follows:
\begin{subequations}
    \begin{align}
    V^\dagger(s)&\coloneqq \sup_\pi V^\pi (s),\quad s\in\mathcal{S},\label{eq:optval-noreg-def}\\
    Q^\dagger(s,a)&\coloneqq \sup_\pi Q^\pi (s,a),\quad s\in\mathcal{S}\ a\in\mathcal{A},
    \end{align}
\end{subequations}
then $\pi^\dagger$ is explicitly written with $Q^*$: 
\begin{equation}
    \pi^\dagger(a|s)=\begin{cases}1&a=\underset{a'\in\mathcal{A}}{\operatorname{argmax}}\ Q^\dagger(s,a'),\\
    0&\mathrm{otherwise}
    \end{cases}\quad s\in\mathcal{S}.
    \label{eq:optpolicy-noreg}
\end{equation}
Thus, maximizing $\mathcal{J}[\pi]$ is replaced by maximizing $V^\pi$ and $Q^\pi$.

One conventional method for finding $V^\pi$ or $Q^\pi$ for given $r$ and $P$ is value iteration, which exploits the fact that $V^\pi$ and $Q^\pi$ are the unique solutions to the following Bellman equations, respectively:
\begin{subequations}
    \begin{align}
        V^\pi(s)&=E_{\pi(a|s)}\left[r(s,a)+E_{P(s'|s,a)}\left[\gamma V^\pi(s')\right]\right],\quad \forall s\in\mathcal{S}\label{eq:bellman-Vnoreg}\\
        Q^\pi(s,a)&=r(s,a)+E_{P(s'|s,a)}\left[\gamma Q^\pi(s',a')\right],\quad \forall (s,a)\in\mathcal{S}\times\mathcal{A}.
       \label{eq:bellman-Qnoreg}
    \end{align}
\end{subequations}
Sometimes, the Bellman operators $\mathcal{B}^\pi:\mathbb{R}^{|\mathcal{S}|}\to\mathbb{R}^{|\mathcal{S}|}$ and $\mathcal{B}^\pi:\mathbb{R}^{|\mathcal{S}|\times |\mathcal{A}|}\to\mathbb{R}^{|\mathcal{S}|^\times |\mathcal{A}|}$ are defined such that Eq.~\eqref{eq:bellman-Vnoreg} and  Eq.~\eqref{eq:bellman-Qnoreg} become $V^\pi(s)=\mathcal{B}^\pi[V^\pi](s)$ and $Q^\pi(s,a)=\mathcal{B}^\pi[Q^\pi](s,a)$, respectively. 
Value iteration finds $V^\pi$ ($Q^\pi$) by a fixed-point iteration of $\mathcal{B}^\pi$ (e.g., $V^{(n+1)}=\mathcal{B}^\pi [V^{(n)}]$, $n=1,2,\ldots$).

In real-world systems, including biological systems, the agent may have a limited access to reward $r$ and state transition $P$ in the environment.
Reinforcement learning (RL) extends to the situation where agents can partially access $r$ and $P$ only through interactions with the environment, and learn $\pi^\dagger$ from time series of states, actions, and rewards $(s_t,a_t,r_t)_{t=0,1,\ldots}$.
One popular method in reinforcement learning is Q-learning, a variant of TD-learning, which uses the fact that $V^\dagger$ and $Q^\dagger$ are the unique solutions to the following Bellman optimality equation:
\begin{subequations}
    \begin{align}
        V^\dagger(s)&=E_{\pi(a|s)}\left[r(s,a)+\gamma\max_{a}E_{P(s'|s,a)}\left[V^\dagger(s')\right]\right],\quad \forall s\in\mathcal{S},\label{eq:bellman-opt-Vnoreg}\\
        Q^\dagger(s,a)&=r(s,a)+\gamma E_{P(s'|s,a)}\max_{a'}\left[Q^\dagger(s',a')\right],\quad \forall(s,a)\in\mathcal{S}\times\mathcal{A}.
    \label{eq:bellman-opt-Qnoreg}
    \end{align}
\end{subequations}
Q-learning approximates the expectation w.r.t. $P$ on RHS of Eq.~\eqref{eq:bellman-opt-Qnoreg} using current samples and relax it with a learning rate $\alpha_t\in(0,1)$:
\begin{subequations}
    \begin{align}
        Q_{t+1}(s_t,a_t)&=(1-\alpha_t)Q_t(s_t,a_t)+\alpha_t\left(r_t+\gamma \max_{a'\in\mathcal{A}}Q_t(s_{t+1},a')\right),\label{eq:Q-learning-update}\\
        Q_{t+1}(s,a)&=Q_{t}(s,a),\quad \forall (s,a)\neq (s_t,a_t).
    \end{align}
    \label{eq:Q-learning}
\end{subequations}
$(s_t,a_t,r_t)_{t=0,1,\ldots}$ is generated by sampling $a_t \sim \pi_b(\cdot | s_t)$, $r_t = r(s_t, a_t)$ and $  s_{t+1} \sim P(\cdot | s_t, a_t)$.
The policy $\pi_b$ followed by action is referred as the behavior policy, which in general differs from the target policy $\pi^\dagger$.
Thus Q-learning is categorized into “off-policy” algorithms.
The update difference $r_t+\gamma \max_a Q_t(s_{t+1},a)-Q_t(s_t,a_t)$ in Eq.~\eqref{eq:Q-learning-update} is referred to as temporal difference (TD) error, which is thought to be represented in neural activity to achieve reinforcement learning in insects~\cite{Bennett2021NatCom} and primates~\cite{Schultz2016NatRevNeu}.

\subsection{Entropy-regularized MDP and RL}
\label{sec:entropy-MDPs}
There are two major issues with the conventional MDP-based RL formulation.
First, approximating the expected value using samples introduces stochastic noise, which is then amplified by the max operator, leading to overestimation and unstable learning.
Second, the resulting deterministic target policy $\pi^\dagger$~\eqref{eq:optpolicy-noreg} lacks a direct theoretical links to exploratory behavior policies $\pi_b$ (e.g., $\varepsilon$-greedy, Boltzmann), introducing an additional parameter tuning.
These issues become more significant when using RL to model biological systems, as biological agents rarely behave deterministically even after learning is complete.
This persistent stochasticity reflects the inherent randomness of biological processes, making deterministic control strategies costly and hence unnatural~\cite{KatoPRR2021}.

The entropy-regularized MDP (ERMDP) addresses these issues by incorporating a control cost into the policy, balancing rewards exploitation and exploration through policy regularization~\cite{Ziebart2010thesis,Haarnoja2017ICML}.
One formulation defines the following objective $\mathcal{J}[\pi]$ for policy $\pi$~\cite{Fox2016}:
\begin{subequations}    \label{eq:cost-discERMDP-w/action}
    \begin{align}
        \mathcal{J}[\pi]&\coloneqq E_{P^\pi[s_{1:\infty},a_{0:\infty}|s_0] p_0(s_0)}\left[\sum_{k=0}^\infty \gamma^k R^\pi(s_k,a_k)\right],\\
        R^\pi(s_k,a_k) &= r(s_k,a_k) -\frac{1}{\beta}\log\frac{\pi(a_k|s_k)}{\pi_\mathrm{ref}(a_k|s_k)},
    \end{align}
\end{subequations}
where the control penalty for $\pi$ is KL divergence from $\pi$ to an intrinsic policy $\pi_\mathrm{ref}:\mathcal{S}\to\Delta(\mathcal{A})$. 
$\beta\in(0,\infty)$ is the inverse control weight and smaller $\beta$ assigns higher control cost.

The formulation with policy regularization also applies when the agent directly controls $P$~\cite{Todorov2009PNAS,Azar2012JMLR}, as adapted in the main text.
Specifically, when the agent controls $P$ itself, at each time $t=0,1,\ldots$, the agent observes state $s_t\in\mathcal{S}$ and selects a policy policy $\pi_t:\mathcal{S}\to\Delta(\mathcal{S})$; and the environment samples next state $s_{t+1}\sim \pi_t(\cdot|s_t)$ and returns $s_{t+1}$ and a reward $r_t$ to the agent.
The objective $\mathcal{J}$ becomes:
\begin{align*}
    \mathcal{J}[\pi]&\coloneqq E_{P^\pi[s_{1:\infty}|s_0] p_0(s_0)}\left[\sum_{t=0}^\infty\gamma^tR^\pi(s_t,s_{t+1})\right],\\
    R^\pi(s_t,s_{t+1})&\coloneqq r(s_t)-\frac{1}{\beta}\log\frac{\pi(s_{t+1}|s_{t})}{p(s_{t+1}|s_t)}, \\
    P^\pi[s_{1:\infty}|s_0]&\coloneqq \prod_{t=0}^\infty \pi(s_{t+1}|s_t)p(s_0),
\end{align*}
for given intrinsic transition $p:\mathcal{S}\to\Delta(\mathcal{S})$.
Note that $\mathcal{J}[\pi]$ is defined only for policies $\pi$ which satisfies $\pi(s'|s)=0$ for all $s,s'\in\mathcal{S}$ such that $p(s'|s)=0$.
The penalty for the transition from time $t$ to $t+1$ is weighted by $\gamma^t$, just as the reward $r(s_{t})$ at time $t$. Alternatively, if we weight the penalty by $\gamma^{t+1}$— aligning it instead with the timing of $r(s_{t+1})$— this is equivalent to replacing $\beta$ by $\beta/\gamma$.
Note that we can also consider the minimization counterpart with negative control penalty, which corresponds to nagative chemotaxis away from chemorepellent~\cite{Tso1974JBac}.

As in standard MDPs, the state value function $V^\pi:\mathcal{S}\to\mathbb{R}$ for $\pi$ (Eq.~\eqref{eq:def-value}) and the optimal state value function $V^*:\mathcal{S}\to\mathbb{R}$ are defined:
\begin{align*}
    V^\pi(s)&\coloneqq 
    E_{P^\pi[s_{0:\infty}|s_0=s]}\left[\sum_{t=0}^\infty\gamma^t R^\pi(s_t,s_{t+1})\right],\quad s\in\mathcal{S},\\
    V^*(s)&\coloneqq \sup_\pi V^\pi(s),\quad s\in\mathcal{S}.
\end{align*}
The optimal policy $\pi^*$ that maximize $\mathcal{J}[\pi]$ is explicitly written using $V^*$ (Eq.~\eqref{eq:EMDP-greedy-policy}): 
\begin{equation*}
    \pi^*(s'|s)=\frac{p(s'|s)\exp(\beta \gamma V^*(s'))}{\sum_{s'\in\mathcal{S}}p(s'|s)\exp(\beta \gamma V^*(s'))},\quad s,s'\in\mathcal{S}.
\end{equation*}
Unlike the optimal policy $\pi^\dagger$ ~\eqref{eq:optpolicy-noreg} in conventional MDP, which relies on max operator, $\pi^*$  instead has softmax operator, resulting in a stochastic and exploratory behavior.
The level of exploration depends on inverse control weight $\beta$ and discount factor $\gamma$.
Similarly, $V^\pi$ is the unique solution to the Bellman equation:
\begin{equation}
    V^\pi(s)=r(s)+E_{\pi(s'|s)}\left[\gamma V^\pi(s)-\frac{1}{\beta}\log\frac{\pi(s'|s)}{p(s'|s)}\right],\quad \forall s\in\mathcal{S}.
    \label{eq:bellman-stationary}
\end{equation}
Bellman operator $\mathcal{B}^\pi:\mathbb{R}^{|\mathcal{S}|}\to\mathbb{R}^{|\mathcal{S}|}$ is defined such that Eq.~\eqref{eq:bellman-stationary} becomes $V^\pi(s)=\mathcal{B}^\pi[V^\pi](s)$ and one can obtain $V^\pi$ by value iteration, a fix-point iteration of $\mathcal{B}^\pi$.

Based on ERMDP, we can consider reinforcement learning problem, referred as entropy-regularized reinforcement learning (ERRL), where the agent only partially access to $r$ through interaction with the environment. 
Similarly to the conventional RL, optimal state value $V^*$ is the unique solution to the Bellman optimality equation:
\begin{subequations}    \label{eq:bellman-optimality}
    \begin{align}
        V^*(s)&=r(s)+\max_\pi E_{\pi(s'|s)}\left[\gamma V^*(s')-\frac{1}{\beta}\log\frac{\pi(s'|s)}{p(s'|s)}\right]\\
        &=r(s)+\frac{1}{\beta}\log \mathcal{Z}[V^*](s),\quad \forall s\in\mathcal{S},
        \label{eq:bellman-opt-G}\\
        \mathcal{Z}[V](s)&\coloneqq \sum_{s'\in\mathcal{S}} p(s'|s)\exp(\beta\gamma V(s')).
    \end{align}
\end{subequations}
Note Eq.~\eqref{eq:bellman-opt-G} is free from max operator.
Similarly to Q-learning~\eqref{eq:Q-learning}, we can consider a variant of TD-learning, called G-learning~\cite{Fox2016}, using Eq.~\eqref{eq:bellman-optimality}:
\begin{subequations}    \label{eq:greedy-G-value}
    \begin{align}
        V_{t+1}(s_t)&=(1-\alpha_t)V_t(s_t)+\alpha_t\left(r_t+\frac{1}{\beta}\log \sum_{s'\in\mathcal{S}}p(s'|s_t)\exp(\beta\gamma V_t(s'))\right),
        \label{eq:greedy-G-value-update}\\
        V_{t+1}(s)&=V_{t}(s),\quad \forall s\neq s_t,
    \end{align}
\end{subequations}
where $\alpha_t\in(0,1)$ is a learning rate and $(s_t,r_t)_{t=0,1,\ldots}$ is sampled by $s_{t+1}\sim\pi_b(\cdot|s_t)$, $r_t=r(s_t)$.

Moreover, ERRL admits a TD-style learning with a behavior policy that is both exploratory and theoretically linked to the optimal policy, thus eliminating the need to adjust parameters during learning.
By leveraging the stochasticity of the optimal policy $\pi^*$, we can use the following “greedy” policy $\pi_t$ as the behavior policy $\pi_b$:
\begin{equation}
    \pi_t(s'|s)\coloneqq\frac{p(s'|s)\exp(\gamma V_t(s'))}{\sum_{s'\in\mathcal{S}}p(s'|s)\exp(\beta\gamma V_t(s'))},\quad s,s'\in\mathcal{S}.
    \label{eq:greedy-G-transition}
\end{equation}
This is equivalent to the softmax (Boltzmann) policy with preference $\beta\gamma V_t$, and it is greedy in a sense that if $V_t(s)=V^*(s)$ for all $s$, then $\pi_t=\pi^*$.
When initialized with a uniform state value, it balances exploration and exploitation: initially it samples a wide range of states, and as learning progresses it increasingly choose states with higher $V_t$.
Greedy G-learning~\cite{Fox2016} is the learning algorithm consisting of the behavior policy~\eqref{eq:greedy-G-transition} and value update~\eqref{eq:greedy-G-value}.
Note that we can also learn a scaled value: $V^*(s)\approx Z_t(s)/\beta$ (Sec.~\ref{sec:distG}), instead of directly learning of value: $V^*(s)\approx V_t(s)$.

Continuous bijection $Z_t(s)\coloneqq\mathcal{F}^{-1}[V_t](s)\coloneqq \exp(\beta V_t(s))$ allows us to obtain another variant of TD-learning.  
Applying the mapping $\mathcal{F}^{-1}$ to the Bellman optimality equation~\eqref{eq:bellman-optimality},
we obtain another variant of Bellman optimality equation~\eqref{eq:bellman-opt-Z} and corresponding optimal policy~\eqref{eq:greedy-opt-Z}:
\begin{subequations}
    \begin{align}
        Z^*(s)&= \exp(\beta r(s))\mathcal{Z}[Z^*](s),\quad\forall s\in\mathcal{S},\label{eq:bellman-opt-Z}\\
        \mathcal{Z}[Z](s)&\coloneqq \sum_{s'\in\mathcal{S}} p(s'|s)(Z(s'))^\gamma,\\
        \pi^*(s'|s)&= \frac{p(s'|s)(Z^*(s'))^\gamma}{\mathcal{Z}[Z^*](s)},\quad s,s'\in\mathcal{S}.\notag\\
        \label{eq:greedy-opt-Z}
    \end{align}
\end{subequations}
One sees inverse control weight $\beta$ is the scaling factor for reward $r$~\cite{Haarnoja2017ICML} in Eq.~\eqref{eq:bellman-opt-Z}, indicating that lower control cost is equivalent to estimating the rewards as higher.
Similarly to greedy G-learning, based on the Bellman optimality equation~\eqref{eq:bellman-opt-Z} and optimal policy~\eqref{eq:greedy-opt-Z}, we obtain a variant of TD-learning~\eqref{eq:greedy-Z-learning} with greedy behavior policy~\eqref{eq:greedy-Z-transition}, called greedy Z-learning~\cite{Todorov2009PNAS}:
\begin{subequations}    \label{eq:greedy-Z-learning}
    \begin{align}
        Z_{t+1}(s_t)&=(1-\alpha_t)Z_t(s_t)+\alpha_t\exp(\beta r(s_t)) \sum_{s'\in\mathcal{S}}p(s'|s_t)Z_t^\gamma(s'),
        \label{eq:greedy-Z-value-update}\\
        Z_{t+1}(s)&=Z_{t}(s),\quad \forall s\neq s_t,
    \end{align}
\end{subequations}
where $\alpha_t\in(0,1)$ is a learning rate and $(s_t,r_t)_{t=0,1,\ldots}$ are sampled by $s_{t+1}\sim\pi_t(\cdot|s_t)$, $r_t=r(s_t)$, and
\begin{equation}
    \pi_t(s'|s)\coloneqq\frac{p(s'|s)Z_t^\gamma(s')}{\sum_{s'\in\mathcal{S}}p(s'|s)Z_t^\gamma(s')},\quad s,s'\in\mathcal{S}.
    \label{eq:greedy-Z-transition}
\end{equation}
\section{Derivation of learning dynamics and continuous limits}
\subsection{Derivation of learning dynamics}
\subsubsection{Derivation of the endogenous cue modulation dynamics}
\label{sec:algorithm-derivation}
In this section, we derive the endogenous cue modulation dynamics~\eqref{eq:dist-Z-learning-diff}, \eqref{eq:dist-G-learning-diff} as an approximated gradient descent on a cost function.

Although multiple choices for a cost function are possible, we specifically choose one based on the following fact: for a continuous bijective mapping $\mathcal{F}$, minimizing the deviation between $V^*$ and $V_t=\mathcal{F}[Z_t]$ is equivalent to minimizing the deviation between the $Z^*\coloneqq\mathcal{F}^{-1}[V^*]$ and $Z_t$, where $Z^*(s)=\mathcal{F}^{-1}[V^*](s)$ for $s\in\mathcal{S}$. 
Following this fact, we consider a simplest dynamics that approximately minimize the following cost function $\mathcal{C}_\mathcal{F}[{Z}_t]$, which incorporates the squared error between the $Z_t$ and $Z^*$ and the physical constraint of diffusion associated with the endogenous cue:
\begin{equation}
    \label{eq:distZ-diff-cost}
        \mathcal{C}_\mathcal{F}[Z_t]\coloneqq \frac{N}{2}\sum_{s\in S}\mu^{\pi_t}(s)(Z^*(s)-Z_t(s))^2-\frac{D}{2\alpha }\sum_{s,s'\in\mathcal{S}} L_{ss'}Z_t(s)Z_t(s'),
\end{equation}
where $\mu^{\pi_t}(s)$ denotes the stationary distribution of $\pi_t$ (assuming $\pi_t$ is ergodic), $D \in \mathbb{R}_+$ represents the diffusion coefficient, $\alpha \in (0,1/N)$ is the time-scaling factor (speed) of the endogenous cue production and degradation by agents, and $L$ denotes the graph Laplacian of the environment $\mathcal{G}$.
The second term of Eq.~\eqref{eq:distZ-diff-cost} models diffusion. 
A larger diffusion coefficient compared to the endogenous cue modulation by agents, the stronger the diffusion becomes, particularly at vertices with fewer agents.

Next, we consider a simplest biologically plausible learning dynamics to minimize the cost function~\eqref{eq:distZ-diff-cost}.
While various learning algorithms, such as offline algorithms~\cite{BartoSutton2018}, are used in engineering applications, not all of these have straightforward biological interpretations.
In this work, we adapt a sequential learning dynamics in which each agent in the population degrades or produces the endogenous cue at each time step.
The simplest sequential dynamics is the gradient descent on $\mathcal{C}_\mathcal{F}[{Z}_t]$; but this requires complete knowledge of the optimal concentration $Z^*$, which is initially unknown to agents.
In RL theory, a method called semi-gradient descent~\cite{BartoSutton2018,precup2000} approximates the gradient by using only the current estimate $Z_t$ and the immediate reward $r(s_t)$.
Following this approach, we approximate the gradient descent on $\mathcal{C}_\mathcal{F}[{Z}_t]$ using the locally sensed current ${Z}_t$ and $r$ by a finite homogeneous population of $N$ agents.
Through the logarithmic mapping $\mathcal{F}$, we derive Eq.~\eqref{eq:dist-Z-learning-diff} as follows:
\begin{subequations}    \label{eq:Zdiff-deriv}
    \begin{alignat}{3}
        Z_{t+1}(s)-Z_{t}(s)&=&&-\alpha \nabla_{Z}\mathcal{C}_\mathcal{F}[Z_t](s)\\
        &=&&-\alpha N\mu^{\pi_t}(s)(Z_t(s)-Z^*(s))-D\sum_{s'\in \mathcal{S}}L_{ss'}Z_t(s')\\
        &\approx&&-\alpha N\sum_{i=1}^N\frac{1}{N}\mathbbm{1}_{s^i_{t}=s}(Z_t(s)-Z^*(s))-D\sum_{s'\in \mathcal{S}}L_{ss'}Z_t(s')\label{eq:Zdiff-step1}\\
        &\approx &&-\alpha\sum_{i=1}^N\mathbbm{1}_{s^i_{t}=s}\left(Z_t(s)-\exp(\beta r(s))\sum_{s'\in\mathcal{S}} p(s'|s)Z_t^\gamma(s')\right)\notag\\
        &&&-D\sum_{s'\in \mathcal{S}}L_{ss'}Z_t(s')\label{eq:Zdiff-step2}\\
        &= &&-\alpha \sum_{i=1}^N\mathbbm{1}_{s^i_{t}=s}\left(Z_t(s)-\rho_t(s^i_{t+1}|s)\exp(\beta r(s))Z_t^\gamma(s^i_{t+1})\right)\notag\\
        &&&-D\sum_{s'\in \mathcal{S}}L_{ss'}Z_t(s')\label{eq:Zdiff-step3},\quad s^i_{t+1}\sim \pi_t(\cdot|s_t^i)\in\mathcal{S},
    \end{alignat}
\end{subequations}
for $s\in\mathcal{S}$, where $\rho_t(s'|s)\coloneqq p(s'|s)/\pi_t(s'|s)$ ($s,s'\in\mathcal{S}$) is the importance sampling ratio~\cite{precup2000} that corrects the deviation between the behavior policy $\pi_t$~\eqref{eq:dist-Z-learning-policy} and the target policy $p$.
In Eq.~\eqref{eq:Zdiff-step1}, we use the stochastic gradient descent~\cite{Sun2018NIPS} by approximating the stationary distribution $\mu^{\pi_t}$ with sample states $(s^i_t)_{i=1}^N$ following a Markov process with policy $\pi_t$.
In Eq.~\eqref{eq:Zdiff-step2}, we replace the unknown optimal value $Z^*$ with the observed reward $(r(s^i_t))_{i=1}^N$ and the current estimate $Z_t$, using the fact that $Z^*$ is the unique solution to the Bellman optimality equation~\eqref{eq:bellman-opt-Z}, as adapted in the semi-gradient descent~\cite{BartoSutton2018,precup2000}: $Z^*=\mathcal{F}^{-1}\circ\mathcal{B}^*\circ\mathcal{F}[Z^*]\approx\mathcal{F}^{-1}\circ\mathcal{B}^*\circ\mathcal{F}[Z_t]$, where $\mathcal{B}[V](s)\coloneqq r(s)+\log\sum_{s'\in\mathcal{S}}p(s'|s)\exp(\gamma V(s'))$ thus $Z^*=\mathcal{F}^{-1}\circ\mathcal{B}^*\circ\mathcal{F}[Z^*]$. 
Eq.~\eqref{eq:Zdiff-step3} is another expression using only sample states $(s_t^i,s_{t+1}^i)$, which is mathematically equivalent to Eq.~\eqref{eq:Zdiff-step2} as long as agents follow the greedy policy $\pi_t$~\eqref{eq:dist-Z-learning-policy}.
We adapt Eq.~\eqref{eq:Zdiff-step2} to simplify the calculation of the continuous limit in Appendix.~\ref{sec:continuous-limit}.

Similarly, through the linear mapping $\mathcal{F}$, we derive Eq.~\eqref{eq:dist-G-learning-diff} by using the associated Bellman optimality equation~\eqref{eq:bellman-opt-G} in Eq.~\eqref{eq:Zdiff-step2} and the greedy policy $\pi_t$~\eqref{eq:dist-G-learning-policy} in Eq.~\eqref{eq:Zdiff-step3}:
\begin{subequations}    \label{eq:Gdiff-deriv}
    \begin{alignat}{3}
        Z_{t+1}(s)-Z_{t}(s)&\approx &&-\alpha\sum_{i=1}^N\mathbbm{1}_{s^i_{t}=s}\left(Z_t(s)-\beta r(s)+\log\sum_{s'\in\mathcal{S}} p(s'|s)\exp(\gamma Z_t(s'))\right)\notag\\
        &&&-D\sum_{s'\in \mathcal{S}}L_{ss'}Z_t(s').\\
        &= &&-\alpha \sum_{i=1}^N\mathbbm{1}_{s^i_{t}=s}\left(Z_t(s)-\beta r(s)+\log \rho_t(s^i_{t+1}|s)\exp(\gamma Z_t(s^i_{t+1}))\right)\notag\\
        &&&-D\sum_{s'\in \mathcal{S}}L_{ss'}Z_t(s'),\quad s^i_{t+1}\sim \pi_t(\cdot|s_t^i)\in\mathcal{S}.
    \end{alignat}
\end{subequations}

\subsubsection{Generalized derivation of learning dynamics}
\label{sec:generalized-derivation}
We can generalize the derivation in Appendix.~\ref{sec:algorithm-derivation} by introducing another continuous bijection $\mathcal{I}:\mathbb{R}^{|\mathcal{S}|}\to\mathbb{R}^{|\mathcal{S}|}$ to explore other learning dynamics.
Minimizing the deviation between $V^*$ and $V_t=\mathcal{F}[Z_t]$ is equivalent to minimizing the deviation between the $\mathcal{I}[Z^*]=\mathcal{I}[\mathcal{F}^{-1}[V^*]]$ and $\mathcal{I}[Z_t]$.
Thus, we consider the following cost function to minimize:
\begin{equation}
    \label{eq:distZ-diff-cost-T}
        \mathcal{C}_{\mathcal{F},\mathcal{I}}[Z_t]\coloneqq \frac{N}{2}\sum_{s\in S}\mu^{\pi_t}(s)(\mathcal{I}[Z^*](s)-\mathcal{I}[Z_t](s))^2-\frac{D}{2\alpha }\sum_{s,s'\in\mathcal{S}} L_{ss'}Z_t(s)Z_t(s').
\end{equation}
Similarly to Appendix.~\ref{sec:algorithm-derivation}, applying the approximation $\mathcal{I}[Z^*]=\mathcal{I}\circ\mathcal{F}^{-1}\circ\mathcal{B}^*\circ\mathcal{F}[Z^*]\approx\mathcal{I}\circ\mathcal{F}^{-1}\circ\mathcal{B}^*\circ\mathcal{F}[Z_t]$ to the gradient descent on the cost~\eqref{eq:distZ-diff-cost-T}, we obtain corresponding generalized learning dynamics.
For example, choosing $\mathcal{F}[Z](s)=\log(Z(s)/\beta)$ and $\mathcal{I}[Z](s)=(Z(s))^\gamma$ gives: 
\begin{alignat*}{3}
        Z_{t+1}(s)-Z_{t}(s)&\approx &&-\alpha\sum_{i=1}^N\mathbbm{1}_{s^i_{t}=s}\left(Z^\gamma_t(s)-\exp(\beta\gamma r(s))\left(\sum_{s'\in\mathcal{S}} p(s'|s)Z^\gamma_t(s')\right)^\gamma\right)\notag\\
        &&&-D\sum_{s'\in \mathcal{S}}L_{ss'}Z_t(s'),
\end{alignat*}
whose continuous limit has a degenerative degradation $H(\mu,Z^\gamma)=\alpha \mu Z^\gamma$~\cite{Tindall2008BullMathBiol}.

Similarly to Sec.~\ref{sec:distZ}, we can generalize the derivation of policy as:
\begin{equation}
    \pi_t(s'|s)\propto p(s'|s)\exp(\mathcal{F}[Z_t](s')).
\end{equation}
As another example, consider a Hill type mapping $\mathcal{I}=\mathcal{F}:[0,\infty)^{|\mathcal{S}|}\to[0,1)^{|\mathcal{S}|}$ defined by $\mathcal{F}[Z](s)\coloneqq (Z(s))^a/(K+(Z(s))^a)\in[0,1)$, $a>0$. 
This nonlinear response, derived from the underlying receptor kinetics, closely matches the experimentally observed response to endogenous concentration~\cite{Fu2018NatCom,Cremer2019Nat,Phan2024PNAS,Bacon1965JMolBiol}.
Note that we must assume $0\leq r(s)< 1-\gamma$ to ensure $V^*\in[0,1)$ since $\min_{s\in\mathcal{S}}r(s)\leq V(s)\leq \max_{s\in\mathcal{S}}r(s)/(1-\gamma)$, $s\in\mathcal{S}$.
For the Michaelis-Menten case $a=1$, the learning dynamics become:
\begin{align}
    \pi_t(s'|s)&\propto p(s'|s)\exp\left(\beta\gamma \frac{Z_t(s')}{K+Z_t(s')}\right),\\
    Z_{t+1}(s)-Z_t(s) &= \sum_{i=1}^N\mathbbm{1}_{s=s^i_t}\alpha\Delta Z_t(s)+D\sum_{s'\in\mathcal{S}}L_{ss'}Z_t(s'),\\
    \Delta Z_t(s)&= -\frac{Z_t(s)}{K+Z_t(s)}+\left(r(s)+\frac1\beta\log\sum_{s'} p(s'|s)\exp\left(\beta\gamma \frac{Z_t(s)}{K+Z_t(s)}\right)\right).
\end{align}
A formal continuous limit of this dynamics has nonlinear chemotactic response $\psi(Z,\nabla Z)=\frac{\gamma\sigma^2 K}{(K+Z)^2}\nabla Z$ and degradation $H(\mu,Z)=\alpha(1-\gamma) \mu\frac{Z}{K+Z}$, both consistent with phenomenologically established functions~\cite{Tindall2008BullMathBiol}.

These generalization are motivated by the interpretation that TD-learning algorithms derived by semi-gradient descent are stochastic approximation algorithms to find the unique fixed point of the Bellman optimality equation~\cite{Borkar2008}.
While we do not formally prove convergence of the generalized algorithms, we anticipate they decrease the deviation between $V_t$ and $V^*$ if a suitable Lyapunov function exists that vanishes only at the mapped fixed point.

\subsection{Derivation of the continuous limits}
\label{sec:continuous-limit}
In this section, we derive the continuous limits~\eqref{eq:log-continuous} and \eqref{eq:lin-continuous} from the coupled Eqs.~\eqref{eq:dist-Z-learning-policy}, \eqref{eq:dist-Z-learning-diff} and Eqs.~\eqref{eq:dist-G-learning-policy}, \eqref{eq:dist-G-learning-diff} by restricting to the one-dimensional Euclidean space for simplicity.
As a discrete approximation of the one-dimensional Euclidean space, we consider a one-dimensional lattice $S^h\coloneqq \{0,\pm h,\pm 2h,\ldots\}$ with a uniform spacing $h\in\mathbb{R}_{>0}$, and define a learning dynamics on the state space $\mathcal{S}=S^h$.
We additionally assume the intrinsic transition probability $p^h$ on $S^h$ as a lazy random walk, similar to Sec.~\eqref{sec:numexp}:
\begin{subequations}    \label{eq:lazy-rw}
    \begin{align}
        p^h(x|x)&=\varepsilon^h(x)=O(h),\\
        p^h(x\pm h|x)&=\frac{1-\varepsilon^h(x)}{2},
    \end{align}
\end{subequations}
where we assume that $\varepsilon^h(x)\to 0$ as $h\to 0$, indicating that the smaller the spacing $h$, the less likely it is to stay at the same vertex.
Although the strict convergence requires taking the continuous limit of both two dynamics simultaneously, we present each continuous limit separately here.   

\subsubsection{Continuous limit of agents' motility}
In this subsection, we roughly derive the continuous limit of agents' motility. 
Given the endogenous cue concentration $Z(t',\cdot):\mathbb{R}\to\mathbb{R}_{\geq 0}$
, we firstly consider agents with a policy~\eqref{eq:dist-G-learning-policy} dependent on $Z$ on $S^h$.
The trajectory of an agent's position is a Markov chain $(\xi^h_n)_{n=0,1,\ldots}$ on $S^h$ with the following transition probability $\pi^h:S^h\to S^h$: 
\begin{subequations}
    \label{eq:discrete-markov}
    \begin{align}
        \pi^h(x|x)&=\varepsilon^h(x)\exp(\gamma Z(t',x))/\mathcal{Z}^h(x)\\
        \pi^h(x\pm h|x)&=\frac{1-\varepsilon^h(x)}{2}\exp(\gamma Z(t',x\pm h))/\mathcal{Z}^h(x),\\
        \mathcal{Z}^h(x)&\coloneqq \varepsilon^h(x)\exp(\gamma Z(t',x))+\frac{1-\varepsilon^h(x)}{2}\left(\exp(\gamma Z(t',x+h))+\exp(\gamma Z(t',x- h))\right),
    \end{align}
\end{subequations}
for $x\in S^h$, where $Z(t',\cdot)$ is assumed to be constant over times $n=0,1,\ldots$.

We define the time step $\Delta t^h\coloneqq h^2/\sigma^2$ depending on spacing $h$, for $n=0,1,\ldots$, the position increment $\Delta \xi^h_n\coloneqq \xi^h_n-\xi^h_{n-1}$ satisfies local consistency~\cite{kushner2001numerical}:
\begin{subequations}
    \begin{align}
        E[\Delta\xi^h_n]&=h\pi^h(x+h|x)-h\pi^h(x-h|x)\\
        &=\gamma h^2\pt_x Z(t',x)+O(h^4)\label{eq:local-consistency-mean-2}\\
        &=\gamma \sigma^2 \pt_x Z(t',x)\Delta t^h+O\left((\Delta {t^h})^2\right), 
        \label{eq:local-consistency-mean}
    \end{align}
\end{subequations}
and 
\begin{subequations}
    \begin{align}
        E\left[\left(\Delta\xi^h_n-E[\Delta\xi^h_n]\right)^2\right]&=E[(\Delta\xi^h_n)^2]-\left(E[\Delta\xi^h_n]\right)^2\\
        &=h^2+O(h^4)\label{eq:local-consistency-var-2}\\
        &=\sigma^2\Delta t^h+O\left((\Delta {t^h})^2\right),
        \label{eq:local-consistency-var}
    \end{align}
\end{subequations}
where in Eqs.~\eqref{eq:local-consistency-mean-2}, \eqref{eq:local-consistency-var-2}, the following Taylor series of $\pi^h$ w.r.t. $h$ is used:
\begin{subequations}
    \begin{alignat}{3}
        \pi^h(x\pm h|x)&=\frac{1-\varepsilon^h(x)}{2}\exp(\gamma Z(t',x))&&\left(1\pm\gamma h\pt_xZ(t',x)+O(h^2)\right)\notag\\
        &&&\left(\exp(\gamma Z(t',x))(1+O(h^2))\right)^{-1}\label{eq:policy-taylor}\\
        &=\frac{1}{2}-\frac{\varepsilon^h(x)}{2}\pm\frac{\gamma}{2} h\pt_xZ(t',&&x)+O(h^2),
    \end{alignat}
\end{subequations}
where in Eq.~\eqref{eq:policy-taylor}, the Taylor series w.r.t. $h$ are performed for $\exp(\gamma Z(t',x\pm h))$ and $\mathcal{Z}^h(x)$.

As $\Delta t^h\to0$ $(h\to 0)$ and local consistency~\eqref{eq:local-consistency-mean} and \eqref{eq:local-consistency-var} is satisfied, if the Markov chain $(\xi^h_n)_{n=0,\ldots}$ satisfies the initial condition $\xi^h_0=x$, the solution $\xi^h(t)\coloneqq\xi_n^h,\ t\in[t_n^h,t_{n+1}^h),\ t^h_n\coloneqq t'+\sum_{i=0}^{n-1}\Delta t^h$ of the corresponding time-continous Markov chain converges to the solution of the following stochastic differential equations~\cite[Theorem 4.1]{kushner2001numerical}:
\begin{equation}
    \label{eq:limit-sde}
    \dd{x}(t)=\gamma \sigma^2 \pt_xV(t,x)\dd{t}+\sigma \dd{W}(t)
\end{equation}
with $x(t')=x$, 
where $W(t)$ is a one-dimensional standard Wiener process.
This indicates that an agent following a biased random walk weighted by the exponent of the surrounding endogenous cue concentration corresponds to an agent following a random walk climbing the gradient of the endogenous cue concentration in the space-time continuous limit.

The Fokker-Planck-Kolmogorov forward (FPK) equation for the probability density $\rho(t,x)$ of random variable $x$ following SDE~\eqref{eq:limit-sde} is given~\cite{Oksendal1992} by
\begin{equation}
        \pt_t\rho(t,x)=-\nabla\cdot(\gamma \sigma^2\nabla Z(t,x)\cdot\rho(t,x))+\frac{\sigma^2}{2}\nabla^2\rho(t,x)
        \label{eq:limit-fpk}
\end{equation}
with $\rho(t',x)=1$, yielding $\phi(\nabla Z,\nabla \log Z)=\gamma \sigma^2 \nabla Z$ (Eq.~\eqref{eq:lin-continuous}).
This indicates that the distribution of agents following a random walk weighted by the exponent of the surrounding endogenous cue concentration~\eqref{eq:dist-Z-learning-policy} corresponds to a diffusion-advection equation with the gradient of the concentration as drift in the space-time continuous limit~\cite{Grima2004PRE,Meyer2023PRE}.

The FPK equation for policy~\eqref{eq:dist-Z-learning-policy} is derived similarly by setting $Z(t,x)\to \log Z(t,x)$: 
\begin{equation}
    \label{eq:limit-fpk-2}
    \pt_t\rho(t,x)=-\nabla\cdot(\gamma \sigma^2\nabla \log Z(t,x)\cdot\rho(t,x))+\frac{\sigma^2}{2}\nabla^2\rho(t,x)
\end{equation}
with $\rho(t',x)=1$, 
yielding $\phi(\nabla Z,\nabla \log Z)=\gamma \sigma^2 \nabla \log Z$ (Eq.~\eqref{eq:log-continuous}).
This indicates that the distribution of agents following a random walk weighted by the power of the surrounding concentration~\eqref{eq:dist-G-learning-policy} corresponds to a diffusion-advection equation with the logarithmic gradient as drift in the limit.

\subsubsection{Continuous limit of the endogenous cue modulation}
First, we consider the modulation of the endogenous cue concentration $Z^h_t:\mathbb{R}\to\mathbb{R}_{\geq 0}$ following Eq.~\eqref{eq:dist-Z-learning-diff} on $S^h\subset \mathbb{R}$.
Here, $Z^h_t(x)=Z^h_t(nh)$, $x\in[nh,(n+1)h)$, $n=0,\pm 1,\ldots$ and set appropriate boundary conditions without strict handling.
Additionally, we approximate the number of agents $\mu_t(x)$ at vertex $x\in S^h$ at time $t$, which is a random variable, by a deterministic value $\mu_t^h(x)$, assuming an infinite number of agents.

With the time step $\Delta t^h=h^2/\sigma^2$ and the learning rate $\alpha^h(x)=\alpha \Delta t^h$, for $x\in S^h$, the update of $Z^h_t(x)$ following Eq.~\eqref{eq:dist-Z-learning-diff} is give by:
\begin{subequations}
    \begin{align}
        Z^h_{t+\Delta t^h}(x)-Z^h_t(x)&=-\alpha^h(x)\mu^h_t(x)\left(Z^h_t(x)-\exp(\beta r(x))\left(\sum_{y\in\{x,x\pm h\}}p^h(y|x)(Z^h_t(y))^\gamma\right)\right)\notag\\
        &+D\left(Z^h_t(x+h)-Z^h_t(x)+Z^h_t(x-h)-Z^h_t(x)\right)
        \label{eq:discrete-dist-Z-leaning-diff-def}\\
        &=-\alpha \Delta t^h\mu^h_t(x)\left(Z^h_t(x)-\exp(\beta r(x))\left((Z^h_t(x))^\gamma+O(h^2)\right)\right)\notag\\
        &+D\sigma^2\Delta t^h
        \frac{1}{h}\left(\frac{Z^h_t(x+h)-Z^h_t(x)}{h}-\frac{Z^h_t(x)-Z^h_t(x-h)}{h}\right)
        \label{eq:discrete-dist-Z-leaning-diff}
    \end{align}
\end{subequations}
By taking the limit as $h\to 0$, the formal space-time continuous limit $Z$ of $Z^h_t$ satisfies the following partial differential equation with appropriate boundary condition: 
\begin{equation}
\label{eq:continuous-dist-Z-leaning-diff}
    \pt_t Z(t,x)=-\alpha \mu(t,x) Z(t,x)+\alpha\mu(t,x)\exp(\beta r(x))Z^\gamma(t,x)+D\sigma^2\nabla^2 Z(t,x),
\end{equation}
where, $\mu(t,x)$ denotes the spatial limit of $\mu^h_t(x)$.
When $\mu(t,x)=N\rho(t,x)$ with $\rho(t,x)$ following~\eqref{eq:limit-fpk}, we obtain $H(\mu,Z)=\alpha\mu Z$, $G(\mu,Z,r)=\alpha\mu \exp(\beta r)Z^\gamma$ (Eq.~\eqref{eq:log-continuous}).
This indicates that the distributed learning dynamics by agent population, based on immediate local concentrations and logarithmic mapping, corresponds to a diffusion equation with linear degradation and nonlinear production in the space-time limit.


Similarly, the update of $Z^h_t(x)$ following Eq.~\eqref{eq:dist-G-learning-diff} with Talor series w.r.t. $h$ is given by:
\begin{subequations}
    \begin{align}
        Z^h_{t+\Delta t^h}(x)-Z^h_t(x)&=-\alpha^h(x)\mu^h_t(x)\left(Z^h_t(x)-\beta r(x)-\left(\ln\sum_{y\in\{x,x\pm h\}}p^h(y|x)\exp(\gamma Z^h_t(y))\right)\right)\notag\\
        &+D\left(Z^h_t(x+h)-Z^h_t(x)+Z^h_t(x-h)-Z^h_t(x)\right)\label{eq:discrete-dist-G-leaning-diff-def}\\
        &=-\alpha^h(x)\mu^h_t(x)\left(Z^h_t(x)-\beta r(x)-\gamma Z^h_t(x)+O(h^2)\right)\notag\\
        &+D\left(Z^h_t(x+h)-Z^h_t(x)+Z^h_t(x-h)-Z^h_t(x)\right),
        \label{eq:discrete-dist-G-leaning-diff}
    \end{align}
\end{subequations}
and the formal space-time continuous limit $Z$ of $Z^h_t$ satisfies the following PDE with appropriate boundary condition:
\begin{equation}
\label{eq:continuous-dist-G-leaning-diff}
    \pt_t Z(t,x)=-\alpha(1-\gamma) \mu(t,x) Z(t,x)+\alpha\beta\mu(t,x) r(x)+D\sigma^2\nabla^2 Z(t,x),
\end{equation}
yielding $H(\mu,Z)=\alpha(1-\gamma)\mu Z$, $G(\mu,Z,r)=\alpha\beta\mu r$ (Eq.~\eqref{eq:lin-continuous}).
This indicates that the distributed learning dynamics by agent population, based on immediate local concentrations and linear mapping, corresponds to a diffusion equation with linear degradation and linear production in the space-time limit.

\section{Supplemental experiments}
This section presents supplemental experiments that reinforce the finding discussed in Sec.~\ref{sec:numexp} and provide additional insights.
\subsection{Maze solving with other couplings}
\label{sec:demostration-supplement}
In this subsection, we show maze-solving dynamics using other couplings: the optimal lin-lin coupling and the non-optimal lin-exp coupling within the same environment as Sec.~\ref{sec:demostration}.

The optimal lin-lin coupling dynamics~\eqref{eq:dist-G-learning-policy} and \eqref{eq:dist-G-learning-diff} show qualitatively similar exploration behavior to the optimal log-exp coupling dynamics~\eqref{eq:dist-Z-learning-policy} and \eqref{eq:dist-Z-learning-diff} (Fig.~\ref{fig:distZ-diff_maze_onetar_VV-VZ} (a) and (b)).
At the almost steady states, the distribution of agents $\mu_t$ and of time-averaged rewards $R^i$ are nearly identical to those of the log-exp coupling (Fig.~\ref{fig:distZ-diff_maze_onetar_VV-VZ} (b) and (c) $t=15001$), though the concentration gradient is more gradual than the log-exp coupling (Fig.~\ref{fig:distZ-diff_maze_onetar_VV-VZ} (a)), consistent with observations in a chain graph (Fig.~\ref{fig:comparison-steady} (b)).
This similarly supports the theoretical equivalence of two optimal log-exp and lin-lin couplings in solving the optimization problem, as verified in Sec.~\ref{sec:optimality}. 

In contrast, the non-optimal lin-exp coupling~\eqref{eq:dist-G-learning-policy} and \eqref{eq:dist-Z-learning-diff} exhibits distinct behavior in both transient and steady states, diverging from the optimal couplings, even from the non-optimal log-lin coupling, .
During transient states, the agents produce the endogenous cue faster than agent following the optimal log-exp coupling, despite both following an exponential production dynamics (Fig.~\ref{fig:distZ-diff_maze_onetar_VV-VZ} (d) $t=1001$).
In the almost steady state, agents generate a high, gradual gradient (Fig.~\ref{fig:distZ-diff_maze_onetar_VV-VZ} (d) $t=15001$), strongly accumulate at the goal (Fig.~\ref{fig:distZ-diff_maze_onetar_VV-VZ} (e) $t=15001$), resulting in comparable $R^i$ values to those of the optimal couplings (Fig.~\ref{fig:distZ-diff_maze_onetar_VV-VZ} (f) $t=15001$).
This is consistent with the behavior observed in a chain graph (Fig.~\ref{fig:comparison-steady} (b), (c)).
Moreover, unlike the other three couplings, not all values of $Z_t$ align with the distance from the goal, implying incomplete exploration in the steady state.

This behavior can be attributed to oversensitive in agent dynamics, as discussed in Sec.~\ref{sec:optimality}.
Even with few agents at the goal, they remain there with high probability, maintaining cue production despite its diffusion. 
Meanwhile, agents that have not reached the goal are attracted by the diffused cue, reaching the goal before fully degrading the cue at unexplored vertices.
As a result, agents aggregate at the goal without sufficient exploration, achieving $R^i$ values comparable to those of a deterministic policy toward the goal.

The non-optimal lin-exp coupling exhibits comparable performance to the optimal couplings in terms of the time-averaged reward $R^i$, due to its tendency to quickly accumulate agents at the goal.
However, this advantage comes with a trade-off in temporal control costs.
For instance, agents newly placed far from the goal incur high control costs due to its oversensitivity but seldom receive rewards, resulting a lower expected discounted cumulative reward $V^\pi$ (see Fig.~\ref{fig:comparison-steady-supplement} (a)).
This distinction highlights the difference between the two metrics: while the time-averaged reward $R^i$ places greater emphasis on the distant future, valuing steady-state accumulation, while the discounted cumulative reward $V^\pi$ prioritizes the near future. 
This highlights a trade-off: optimizing for $R^i$ requires agents to tolerate higher control penalties initially to secure rewards in the long run, whereas optimizing for $V^\pi$ emphasizes minimizing immediate control costs, balancing them with achievable near-term rewards.  
This study used the time-averaged reward metric to compare agents histories across time steps, but the choice of metric should align with biological contexts.
For example, in environments with sparse targets, simple agents like neutrophils may prioritize low control costs for near-term rewards, making the discounted metric more suitable. 

These results reinforce 
the importance of optimal coupling between cue sensing and modulation for efficient migration, as discussed in Sec.~\ref{sec:demostration} and \ref{sec:optimality}.
The results also suggest that while optimization-based modeling effectively clarifies how different components of complex dynamics contribute to function,  selecting the optimization criteria, such as $R^i$ or $V^\pi$, should reflect specific biological contexts to ensure greater relevance.

\begin{figure*}[tbp]
    \hfill
    \begin{minipage}[b]{\linewidth}
        \centering
        \includegraphics[width=\textwidth]{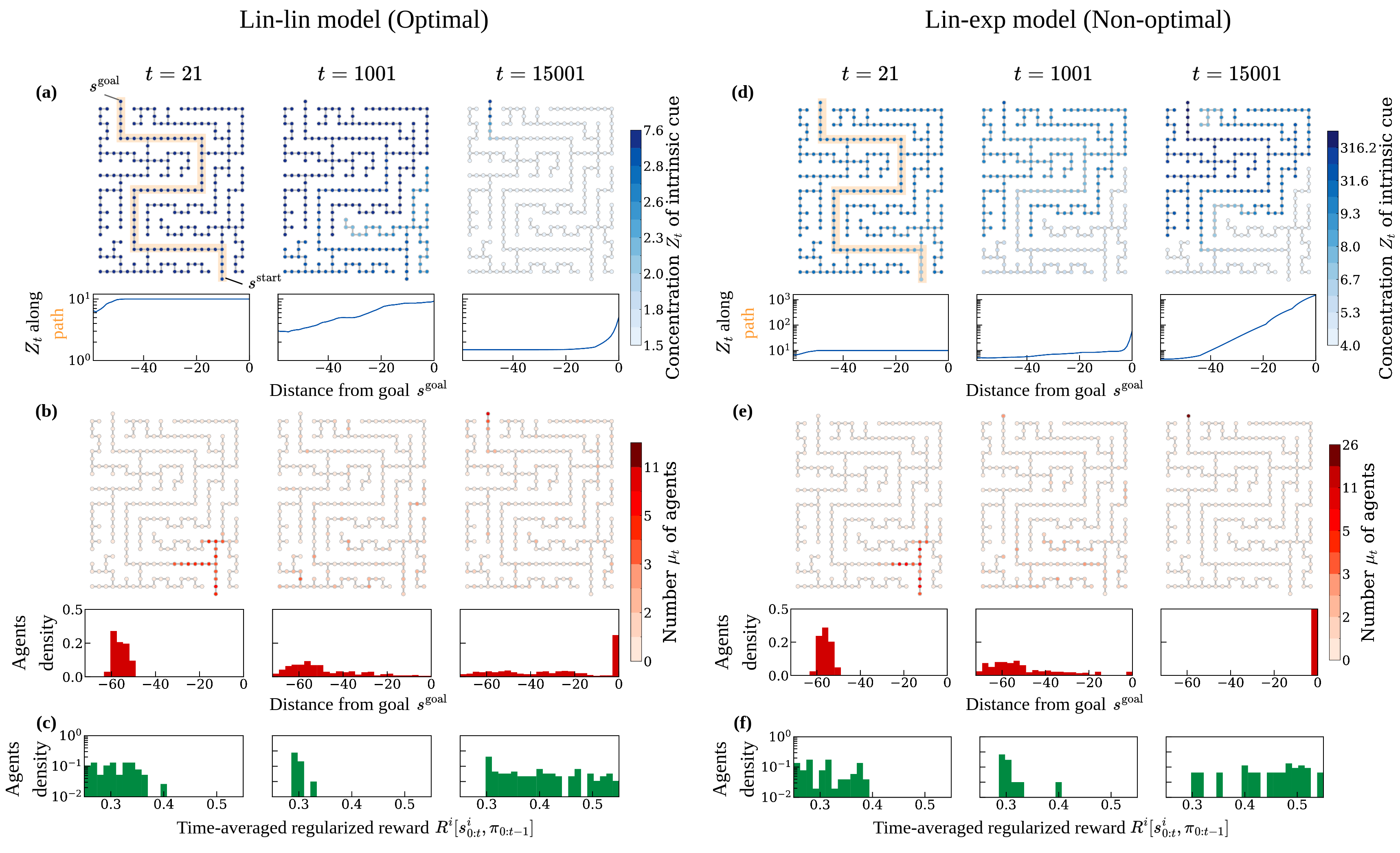}
    \end{minipage}
    \hfill
    \caption{
    Snapshots at $t=21,\, 1001,\, 15001$ for one trial of the maze exploration by the optimal lin-lin coupling dynamics~\eqref{eq:dist-G-learning-policy}, \eqref{eq:dist-G-learning-diff} and the non-optimal lin-exp coupling model~\eqref{eq:dist-G-learning-policy}, \eqref{eq:dist-Z-learning-diff}. 
    The distributions of the endogenous cue $Z_t$ are shown on the maze (upper panel) and on the path from $s^\mathrm{start}$ to $s^\mathrm{goal}$ (lower panel) for the lin-lin (a) and the lin-exp (d) couplings. 
    The path is highlighted with orange background in the maze and the distribution along the path is shown as a function from the goal $s^\mathrm{goal}$.
    The agent distribution $\mu_t:\mathcal{S}\to\{0,1,\ldots,N\}$ defined as $\mu_t(s)\coloneqq\sum_{i=1}^N\mathbbm{1}_{s_t^i=s}$ is similarly shown for the lin-lin (b) and the lin-exp (e) couplings. 
    Time-averaged regularized reward distribution $(R^i[s^i_{0:t},\pi_{0:t-1}])^N_{i=1}$ for the lin-lin (c) and the lin-exp (f) couplings.
    The vertices are color-coded by the value of cue concentration in (a) and (c) and that of agent density in (b) and (e).
    The parameter values used are $\alpha=0.0098$, $N=100$, $D=0.01$, $\gamma=0.8$, $Z_0=10.0$, $\varepsilon=0.5$, $\beta$, $r_\mathrm{target}=1.5$, $r_\mathrm{default}=0.3$.
    }
    \label{fig:distZ-diff_maze_onetar_VV-VZ}
\end{figure*}

\subsection{Supplementary information of state value to Figure~\ref{fig:comparison-steady-supplement}}
\label{sec:optimality-supp}
In this section, we show additional metrics for $V^\pi(s^0)$ (Fig.~\ref{fig:comparison-steady-supplement} (a) and (b)) and the average of $V^\pi$ over $\mathcal{S}$ (Fig.~\ref{fig:comparison-steady-supplement} (c) and (d)), supplementing the analysis of $V^\pi(s^{19})$ (Fig.~\ref{fig:comparison-steady}) in Sec.~\ref{sec:optimality}.
Due to the monotonic distribution of endogenous cue $r$, the average lies between $V^\pi(s^0)$ and $V^\pi(s^{19})$, thus detailed description of the average is omitted here. 

\begin{figure}[tbp]
    \centering
    \includegraphics[width=0.70\textwidth]{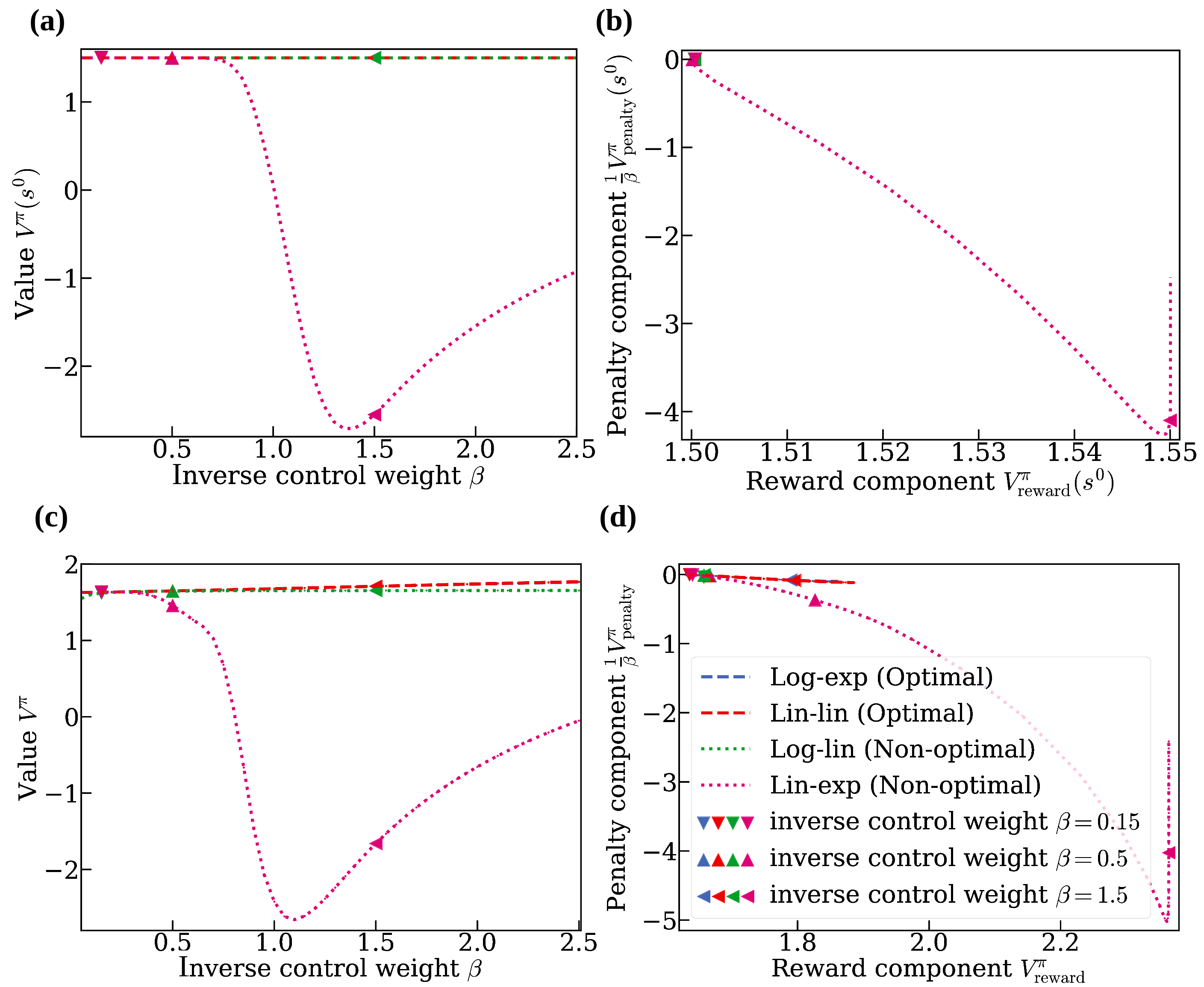}
    \caption{
    State value $V^\pi$ at vertex $s^{19}$ (a) and its average over $\mathcal{S}$ (c) for each policy $\pi[\tilde{Z}_T]$ based on the trial-averaged steady concentration $\tilde{Z}_T$ for the four couplings with $\beta\in[0.25,2.5]$. 
    The relationship between the reward component $V^\pi_\mathrm{reward}$ and penalty component $\frac{1}{\beta}V^\pi_\mathrm{penalty}$ of $V^\pi(s^{19})$ (b), and of average of $V^\pi$ $\mathcal{S}$ (d), for the four couplings with $\beta\in[0.25,2.5]$.
    The other parameters values are same as Fig.~\ref{fig:comparison-steady}.
    }
    \label{fig:comparison-steady-supplement}
\end{figure}

We observe similar trends at $s^0$ to those at $s^{19}$.
The optimal log-exp and lin-lin couplings exhibits almost identical and higher values $V^\pi$ compared to the non-optimal log-lin and lin-exp couplings (Fig.~\ref{fig:comparison-steady-supplement} (a)).
Also, the optimal couplings achieve a good balance between  the reward component $V^\pi_\mathrm{reward}$ the penalty component $V^\pi_\mathrm{penalty}/\beta$ (Fig.~\ref{fig:comparison-steady-supplement} (b)).
Additionally, the non-optimal couplings show lower $V^\pi$ and an imbalance between these components.
This consistency at both $s^0$ and $s^{19}$ reinforces the robustness of these trends across vertices.

The deficiency of the non-optimal couplings is more pronounced at $s^{0}$ than at $s^{19}$, especially for the non-optimal lin-exp coupling: for large $\beta$, the lin-exp coupling incurs a much higher penalty $V^\pi_\mathrm{penalty}(s^0)$  (Fig.~\ref{fig:comparison-steady-supplement} (a)), resulting in a significantly lower $V^\pi(s^0)$ (Fig.~\ref{fig:comparison-steady-supplement} (b)).
This discrepancy arises because $V^\pi_\mathrm{reward}(s)$ places more emphasis on rewards near the starting vertex $s$, which diminish as $s$ are farther from the goal $s^{19}$. 
Consequently, the strong control by the lin-exp couplings around $s^0$ does not contribute to reward gain and instead leads to high penalties.   


These findings support the conclusion provided in Sec.~\ref{sec:optimality} that the non-optimal couplings are indeed inferior in terms of optimization for target accumulation, reflecting their qualitative difference in behavior. 

\section{Convergence analysis on an approximated model}
In this section, we discuss a convergence analysis of the lin-lin dynamics by assuming a time-scale separation between the policy and cue modulation.
In Sec.~\ref{sec:conv-approx-model}, we introduce an approximated model for convergence analysis.
Then, we provide a convergence result and discussion in Sec.~\ref{sec:conv-result}. Its proof is provided in Sec.~\ref{sec:conv-proof}. 

\subsection{An approximated model for convergence analysis}
\label{sec:conv-approx-model}
For an ease of a convergence analysis, we hereafter focus on an approximated version~\eqref{eq:original-evol-dist-discrete}, \eqref{eq:original-evol-chem-discrete} of the lin-lin model~\eqref{eq:dist-G-learning-policy}, \eqref{eq:dist-G-learning-diff}:
\begin{subequations}
    \begin{align}
        \mu_{t+\Delta t}(s)-\mu_t(s)=\rho_t\Delta t\sum_{s'\in\mathcal{N}(s)}A[Z](s'|s)\mu_t(s'),\quad s\in\mathcal{S},
        \label{eq:original-evol-dist-discrete}\\
        Z_{t+\Delta t}(s)=\alpha_t\Delta t\left(-\mu_t(s)\Delta Z(s) -D'\sum_{s'\in\mathcal{S}}L_{ss'}Z(s')\right),\quad s\in\mathcal{S},
        \label{eq:original-evol-chem-discrete}        
    \end{align}
\end{subequations}
where $\alpha_t\in(0,1),\rho_t\in(0,1]$ are the step sizes, $\Delta Z(s)\coloneqq Z(s)-\beta r(s)-\log\sum_{s'\in\mathcal{S}}p(s'|s)\exp(\gamma Z(s'))$.
$A:\mathbb{R}^{|\mathcal{S}|}\to\mathbb{R}^{|\mathcal{S}|\times |\mathcal{S}|}$ is the infinitesimal generator (c.f. \cite[Section 11.5]{PutermanMDP1994}) of the original policy $\pi_{\Delta t}$~\eqref{eq:dist-G-learning-policy} with a lazy intrinsic policy $p$ with parameter $\varepsilon^p \Delta t$ (as in Sec.~\ref{sec:numexp}, \ref{sec:continuous-limit}):
\begin{subequations}
    \label{eq:infinitesimal-default}
    \begin{align}
        A[Z](s|s) &\coloneqq \lim_{\Delta t\to0}\frac{\pi_{\Delta t}[Z](s|s)-1}{\Delta t}= -\varepsilon^p\sum_{s''\in\mathcal{N}(s)}\frac{\exp\left(\gamma (Z(s'')-Z(s))\right)}{|\mathcal{N}(s)|},\\
        A[Z](s'|s) &\coloneqq \lim_{\Delta t\to0}\frac{\pi_{\Delta t}[Z](s'|s)}{\Delta t}= \varepsilon^p\frac{\exp\left(\gamma (Z(s')-Z(s))\right)}{|\mathcal{N}(s)|},\quad s'\in\mathcal{N}(s),\\
        A[Z](s'|s) &=0,\quad s\notin\mathcal{N}(s)^\cup\{s\},
    \end{align}
\end{subequations}
for $s\in\mathcal{S}$.
For the agent motility, instead of individual agents' states following policy~\eqref{eq:dist-G-learning-policy}, we focus on the state distribution $\mu_t\in\Delta(\mathcal{S})$ of agents following Eq.~\eqref{eq:original-evol-dist-discrete}.
This equation is the first-order approximation of the evolution of agents distribution following policy~\eqref{eq:dist-G-learning-policy} with respect to $\Delta t$. 
For endogenous cue dynamics, we have $D=N\alpha_t\Delta t D'$ between Eqs.~\eqref{eq:dist-G-learning-diff} and \eqref{eq:original-evol-chem-discrete}.
One sees that each equation of the pair~\eqref{eq:original-evol-dist-discrete} and \eqref{eq:original-evol-chem-discrete} converges to each of the same limiting PDEs~\eqref{eq:lin-continuous}, respectively. 

To establish a convergence analysis, we assume a time-scale separation between Eqs.~\eqref{eq:original-evol-dist-discrete} and \eqref{eq:original-evol-chem-discrete}: $\rho_t/\alpha_t\to 0$ as $t\to\infty$.
Then, the solution $(\mu_t, Z_t)_t$ of Eqs.~\eqref{eq:original-evol-dist-discrete} and \eqref{eq:original-evol-chem-discrete} is described by the following ODEs in the regime where $\varepsilon\to0$ and $t\to\infty$~\cite{Borkar2008}:
\begin{subequations}
    \label{eq:two-timescale-ode}
    \begin{align}
        \dot{\mu}_t&=\mathcal{P}[Z_t,\mu_t],\\
        \dot{Z}_t&=\frac{1}{\varepsilon}\mathcal{T}[Z_t,\mu_t],
    \end{align}
\end{subequations}
where 
\begin{subequations}
    \begin{align}
        \mathcal{P}[Z,\mu](s)&\coloneqq \sum_{s'\in\mathcal{N}(s)}A[Z](s|s')\mu(s')+A[Z](s|s)\mu(s),\quad s\in\mathcal{S},\\
        \mathcal{T}[Z,\mu](s)&\coloneqq-\mu(s)\Delta Z(s)-D'\sum_{s'\in\mathcal{S}}L_{ss'}Z(s'),\quad s\in\mathcal{S}
        \label{eq:modified-evol-chem-operator}
    \end{align}
\end{subequations}
for $s\in\mathcal{S}$.
Intuitively, the endogenous cue concentration $Z_t$ is the fast component while agent distribution $\mu_t$ is the slow (quasi-static) component, which sees the fast cue concentration as equilibrated.
This represents that agents move slowly, while the endogenous cue rapidly are degraded, produced and diffused, which align with several theories~\cite{Zampetaki2021PNAS} and would be plausible to some organisms~\cite{Tweedy2020Science}.

\subsection{Convergence analysis result and discussion}
\label{sec:conv-result}
We obtain the following convergence result of the iteration~\eqref{eq:original-evol-dist-discrete} \eqref{eq:original-evol-chem-discrete} for sufficiently small $D'$ and $\varepsilon^p$ by using two-timescale ODE scheme~\cite{Borkar2008}.
\begin{theorem}
    \label{thm:ODE-convergence}
    Suppose the learning rate $(\rho_t)_t$ and $(\alpha_t)_t$ satisfy
    \[\sum_{t=0}^\infty\rho_t=\sum_{t=0}^\infty\alpha_t=\infty,\quad\sum_{t=0}^\infty|\rho_t|^2+|\alpha_t|^2<\infty,\quad\rho_t/\alpha_t\to0\ (t\to\infty).\]
    Assume there exists $\delta\in(0,1)$ such that $(\mu_t)_{t\leq 0}$ generated by $\dot{\mu}_t=\mathcal{P}[Z^*[{\mu_t}],\mu_t]$ is confined inside 
    \begin{equation}
        \label{eq:conv-assump-policy}
        \tilde{\Delta}_{\delta,D',\mathcal{G}}(\mathcal{S})\coloneqq \left\{\mu\in\Delta(\mathcal{S})\relmiddle|\delta \leq\mu(s)\leq (1-D'\operatorname{maxdeg}(\mathcal{G})),\forall s\in\mathcal{S}\right\}
    \end{equation}
    and 
    \begin{equation}
        \label{eq:conv-assump-parameters}
        2|\mathcal{S}|\varepsilon^p\left(\exp(2M_{\beta,\gamma,r})+\frac{(1-\gamma)\delta+2D'\operatorname{maxdeg}(\mathcal{G})}{\gamma(1-\gamma)^2\delta^2}\frac{(M_{\beta,\gamma,r})^2}{\exp(m_{\beta,\gamma,r})}\right)<1,
    \end{equation}
    where $\operatorname{maxdeg}(\mathcal{G})\coloneqq \max_{s\in\mathcal{S}}\operatorname{deg}(s)$, and $M_{\beta,\gamma,r} \coloneqq \|r\|_\infty\beta\gamma/(1-\gamma)$ and $m_{\beta,\gamma,r}\coloneqq r_\mathrm{min}\beta\gamma/(1-\gamma)$ are constant depends on parameters.
    Then, the discrete trajectory $(\mu_t,Z_t)_t$ generated by the iteration~\eqref{eq:original-evol-dist-discrete}, \eqref{eq:original-evol-chem-discrete} converges to a fixed point $(\mu^*,Z^*[\mu^*])$ as $t\to\infty$.
\end{theorem}
\begin{proof}
    The assumptions of \cite[Chapter 6. Theorem 2]{Borkar2008} are satisfied. 
    This guarantees the convergence.
\end{proof}
Assumptions~\eqref{eq:conv-assump-policy}, \eqref{eq:conv-assump-parameters} on parameters indicates that agents disperse among the entire environments, move at a sufficiently slow speed and degrade or produce the endogenous cue sufficiently faster than its diffusion.
The diminishing property of step sizes (e.g., $\rho_t=1/t^{2/3},\alpha_t=1/t$) is a typical assumption of asymptotic convergence analysis, but has a physical interpretation that the progress of time slows down as time advances.
Even if the step sizes are constant ($\rho_t\equiv\rho\in(0,1)$, $\alpha_t\equiv\alpha\in(0,1)$), we may prove a similar convergence around the fixed point, i.e., $\limsup_{t\to\infty}\|\mu_t-\mu^*\|_\infty=O(\alpha)+O(\rho/\alpha)$ and $\limsup_{t\to\infty}\|Z_t-Z^*[{\mu^*}]\|_\infty=O(\alpha)+O(\rho/\alpha)$~\cite[Chapter 9.4]{Borkar2008}.

An important outcome of the convergence result is the evaluation between the optimal $Z^*$ and the equilibrium $Z^*[\mu^*]$ using the perturbation parameter $D'$.
Optimality of learning dynamics is perturbed by the regularization term, which is added to the cost function $\mathcal{C}_\mathcal{F}$ to model cue diffusion.
This perturbation is characterized by the following inequality of Lemma~\ref{lemma:mod-bellman-cont}:
\[    \|Z^*[\mu]-Z^*\|_\infty\leq \frac{2D'\operatorname{maxdeg}(\mathcal{G})\beta\|r\|_\infty}{(1-\gamma)^2\delta},
\]
which indicates that the equilibrium $Z^*[\mu^*]$ converges to the optimal value $Z^*$ in the limit of no diffusion, i.e., $D'\to 0$ (equivalently $N\to\infty$). 

A limitation of our convergence result is that we cannot check if the assumption~\eqref{eq:conv-assump-policy} holds for given parameters before execution.
Since we can observe the convergence with parameters in much broader region, we expect that we would discard this formal assumption if we could directly evaluate the modulation $\mathcal{T}$ with $\mu_t$ varied.  
Such generalization could allow convergence analysis on the log-exp model, in which we cannot exploit contraction property.

\subsection{Proof of the convergence analysis}
\label{sec:conv-proof}
The main aim of this subsection is to check the condition of Theorem.~\ref{thm:ODE-convergence} by using two-timescale ODE scheme.

First, we can check the Lipschitz continuities of the operators $\mathcal{P}[Z,\mu]$ and $\mathcal{T}[Z,\mu]$.
\begin{proposition}
    \label{prop:P-Lipschitz}
    $\mathcal{P}[Z,\mu]$ is Lipschitz w.r.t. both $Z$ and $\mu$.
\end{proposition}
\begin{proof}
    Let $M_{\beta,\gamma,r} \coloneqq \|r\|_\infty\beta\gamma/(1-\gamma)$ and $m_{\beta,\gamma,r}\coloneqq r_\mathrm{min}\beta\gamma/(1-\gamma)$, then $m_{\beta,\gamma,r}\leq |\gamma Z(s)|\leq M_{\beta,\gamma,r}$.
    Next, we check that $\mathcal{P}[Z,\mu]$ is Lipschitz w.r.t. $Z$:
    \begin{align*}
        \|\mathcal{P}[Z,\mu]-\mathcal{P}[Z',\mu]\|_\infty&\leq \sum_{s\in\mathcal{S}}\left|\sum_{s'\in\mathcal{S}}\varepsilon^p\mu(s')\left(A[Z](s|s')-A[Z'](s|s')\right)\right|\\
        &\leq \sum_{s'\in\mathcal{S}}\varepsilon^p\mu(s')\sum_{s\in\mathcal{S}}\left|A[Z](s|s')-A[Z'](s|s')\right|\\
        &\leq \sum_{s'\in\mathcal{S}}\varepsilon^p\mu(s')\sqrt{|\mathcal{S}|}\|A[Z](\cdot|s')-A[Z'](\cdot|s')\|_2\\
        &\leq \sum_{s'\in\mathcal{S}}\frac{2\sqrt{|\mathcal{S}|}\varepsilon^p\mu(s')}{|\mathcal{N}(s')|}\|\exp(\gamma (Z(\cdot)-Z(s')))-\exp(\gamma (Z'(\cdot)-Z'(s')))\|_2\\
        &\leq \sum_{s'\in\mathcal{S}}\frac{2\sqrt{|\mathcal{S}|}\varepsilon^p\mu(s')}{|\mathcal{N}(s')|\exp(m_{\beta,\gamma,r})}\|\exp(\gamma Z(\cdot))-\exp(\gamma Z'(\cdot))\|_2\\
        &\leq \sum_{s'\in\mathcal{S}}\frac{2\sqrt{|\mathcal{S}|}\varepsilon^p M_{\beta,\gamma,r}\mu(s')}{|\mathcal{N}(s')|\exp(m_{\beta,\gamma,r})}
        \|Z-Z'\|_2\\
        &\leq \sum_{s'\in\mathcal{S}}\frac{2\sqrt{|\mathcal{S}|}\varepsilon^pM_{\beta,\gamma,r}\mu(s')}{\exp(m_{\beta,\gamma,r})}
        \|Z-Z'\|_2\\
        &\leq \sum_{s'\in\mathcal{S}}\frac{2\sqrt{|\mathcal{S}|}\varepsilon^pM_{\beta,\gamma,r}}{\exp(m_{\beta,\gamma,r})}
        \|Z-Z'\|_2\\
        &\leq \frac{2|\mathcal{S}|\varepsilon^pM_{\beta,\gamma,r}}{\exp(m_{\beta,\gamma,r})}
        \|Z-Z'\|_\infty
    \end{align*}
    where we used the definition of $\|\cdot\|_\infty$ in the first inequality, the norm equlivalence between $\|\cdot\|_1$ and $\|\cdot\|_2$ in the third, the definition~\eqref{eq:infinitesimal-default} of $A[Z]$ in the forth, $|x/a-y/b|<|x-y|/c$ for $0<\mathrm{min}\{a,b\}\leq c$ in the fifth, $M_{\beta,\gamma,r}/\beta\gamma$-Lipschitz property of $Z$ in the sixth, $\mathcal{N}(s)\geq 1$ for any $s$ due to the connectivity in the seventh, $\mu(s)\leq 1$ for any $s$ in the eighth, and the norm equlivalence between $\|\cdot\|_2$ and $\|\cdot\|_\infty$ in the last.

    Then, we also check that $\mathcal{P}[Z,\mu]$ is Lipschitz w.r.t. $\mu$:
    \begin{align*}
    \|\mathcal{P}[Z,\mu]-\mathcal{P}[Z,\mu']\|_\infty&\leq
    \|A[Z]\|_\infty \|\mu-\mu'\|_1\\
    &\leq 2\varepsilon^p|\mathcal{S}|\exp\left(2M_{\beta,\gamma,r}\right)\|\mu-\mu'\|_\infty,
    \end{align*}
    where we used the following an upper bound of $\|A[Z]\|_\infty$:
    \begin{align*}
    \|A[Z]\|_\infty&= \varepsilon^p\max_{s\in\mathcal{S}}2\sum_{s'\in\mathcal{N}(s)}\frac{\exp\left(\gamma (Z(s')-Z(s))\right)}{|\mathcal{N}(s)|}\\
    &\leq \varepsilon^p\max_{s\in\mathcal{S}}2\sum_{s'\in\mathcal{N}(s)}\frac{\exp\left(2\gamma\|Z\|_\infty\right)}{|\mathcal{N}(s)|}\\
    &=2\varepsilon^p\exp\left(2M_{\beta,\gamma,r}\right).
    \end{align*}
\end{proof}

\begin{proposition}
    $\mathcal{T}[Z,\mu]$ is Lipschitz w.r.t. both $Z$ and $\mu$.
\end{proposition}
\begin{proof}
    \begin{align*}
        \|\mathcal{T}[Z,\mu]-\mathcal{T}[Z',\mu]\|_\infty&\leq \|\mu\|_\infty \|\Delta Z-\Delta Z'\|_\infty +D'\|L(Z-Z')\|_\infty\\
        &= \|\mu\|_\infty \|Z-\mathcal{B}[Z]-Z'+\mathcal{B}[Z']\|_\infty +D'\|L(Z-Z')\|_\infty\\
        &\leq \|\mu\|_\infty (\|Z-Z'\|_\infty+\|\mathcal{B}[Z]-\mathcal{B}[Z']\|_\infty) +D'\|L(Z-Z')\|_\infty\\
        &\leq \|\mu\|_\infty (1+\gamma)\|Z-Z'\|_\infty+D'\|L(Z-Z')\|_\infty\\
        &\leq \left(\|\mu\|_\infty (1+\gamma)+D'\|L\|_1\right)\|Z-Z'\|_\infty,
    \end{align*}
    where $\mathcal{B}:\mathbb{R}^{|\mathcal{S}|}\to\mathbb{R}^{|\mathcal{S}|}$ is defined as $\mathcal{B}[Z](s)\coloneqq \beta r(s)+\log\sum_{s'\in\mathcal{S}}p(s'|s)\exp(\gamma Z(s'))$,
    and 
    \begin{align*}
        \|\mathcal{T}[Z,\mu]-\mathcal{T}[Z,\mu']\|_\infty&\leq \|\mu-\mu'\|_\infty \|\Delta Z\|_\infty\\
        &\leq \|\mu-\mu'\|_\infty \left((1+\gamma)M_{\beta,\gamma,r}/\gamma+\beta\|r\|_\infty\right),
    \end{align*}
    where we used
    \begin{align*}
        \|\Delta Z\|_\infty &= \max_{s\in\mathcal{S}}\left\{Z(s)-\beta r(s)-\ln\sum_{s'\in\mathcal{S}}p(s'|s)\exp(\gamma Z(s'))\right\}\\
        &\leq \|Z\|_\infty+\beta\|r\|_\infty+\max_{s\in\mathcal{S}}\left\{\ln\sum_{s'\in\mathcal{S}}p(s'|s)\exp(\gamma Z(s'))\right\}\\
        &\leq \|Z\|_\infty+\beta\|r\|_\infty+\max_{s\in\mathcal{S}}\left\{\ln\sum_{s'\in\mathcal{S}}p(s'|s)\exp(\gamma \|Z\|_\infty)\right\}\\
        &= \|Z\|_\infty+\beta\|r\|_\infty+\gamma \|Z\|_\infty\\
        &\leq (1+\gamma)M_{\beta,\gamma,r}/\gamma+\beta\|r\|_\infty
    \end{align*}
\end{proof}

Under some assumption on $\mu$, $\mathcal{G}$ and $D'$, one can prove that $\mathcal{T}[\cdot,\mu]$ has a unique fixed point $Z^*[\mu]$, which converges to $Z^*$ as $D'\to 0$.
\begin{lemma}
    \label{lemma:mod-bellman-cont}
    Assume $\mu$ satisfies $\delta \leq\mu(s)\leq (1-D'\operatorname{maxdeg}(\mathcal{G})),\ \forall s\in\mathcal{S}$ for some $\delta>0$.
    Then, $\mathcal{T}[\cdot,\mu]$ has a unique fixed point $Z^*[\mu]$, which satisfies
    \begin{equation}
        \|Z^*[\mu]-Z^*\|_\infty\leq \frac{2D'\operatorname{maxdeg}(\mathcal{G})\beta\|r\|_\infty}{(1-\gamma)^2\delta}.
        \label{eq:V-mu-bound}
    \end{equation}
\end{lemma}
\begin{proof}
    For $\mu$, $\mathcal{G}$ and $D'$ satisfy the assumption, 
    we define another Bellman opeartor $\mathcal{B}_{\mu}$ by 
    \begin{equation}
    \mathcal{B}_{\mu}[Z] = \operatorname{diag}(\mu)\cdot\mathcal{B}[Z]+(I_{|\mathcal{S}|}-\operatorname{diag}(\mu))\cdot Z-D'L\cdot Z,
    \end{equation}
    similarly to \cite[p.97]{Borkar2008}.
    Then, we have $\mathcal{T}[\cdot,\mu]=\mathcal{B}_{\mu}[Z]-Z$.
    \begin{align*}
        \|\mathcal{B}_{\mu}[Z]-\mathcal{B}_{\mu}[Z']\|_\infty&=
        \|\operatorname{diag}(\mu)(\mathcal{B}[Z]-\mathcal{B}[Z'])+(I_{|\mathcal{S}|}-\operatorname{diag}(\mu)-D'L)(Z-Z')\|_\infty\\
        &\leq
        \|\operatorname{diag}(\mu) \mathbbm{1}_{|\mathcal{S}|}\|\mathcal{B}[Z]-\mathcal{B}[Z']\|_\infty+(I_{|\mathcal{S}|}-\operatorname{diag}(\mu)-D'L)\mathbbm{1}_{|\mathcal{S}|}\|Z-Z'\|_\infty\|_\infty\\
        &\leq
        \|\operatorname{diag}(\mu) \mathbbm{1}_{|\mathcal{S}|}\gamma\|Z-Z'\|_\infty+(I_{|\mathcal{S}|}-\operatorname{diag}(\mu)-D'L)\mathbbm{1}_{|\mathcal{S}|}\|Z-Z'\|_\infty\|_\infty\\
        &=
        \|(I_{|\mathcal{S}|}-(1-\gamma)\operatorname{diag}(\mu)-D'L)\mathbbm{1}_{|\mathcal{S}|}\|Z-Z'\|_\infty\|_\infty\\
        &\leq
        \|\left((1-(1-\gamma)\delta) I_{|\mathcal{S}|}-D'L\right)\mathbbm{1}_{|\mathcal{S}|}\|_\infty\|Z-Z'\|_\infty\\
        &=
        (1-(1-\gamma)\delta)\|Z-Z'\|_\infty,
    \end{align*}
    where we used the fact that $\operatorname{diag}(\mu)$ and $I_{|\mathcal{S}|}-\operatorname{diag}(\mu)-D'L$ are matrices whose all elements are non-negative in the first equality, $\mathcal{B}$ is $\gamma$-contraction mapping in the second, and $\delta I_{|\mathcal{S}|}\leq\operatorname{diag}(\mu)$ in the third.
    Thus, $\mathcal{T}[\cdot,\mu]$ is $(1-(1-\gamma)\delta)$-contraction and has a unique fixed point $Z^*[\mu]\in\mathbb{R}^{|\mathcal{S}|}$ by the fixed point theorem.

    Moreover, we obtain Eq.~\eqref{eq:V-mu-bound} by $\|Z^*\|_\infty\leq \beta\|r\|_\infty/(1-\gamma)$ and
    \begin{align*}
        \|Z^*[\mu]-Z^*\|_\infty&=\|\mathcal{B}_{\mu}[Z^*[\mu]]-\mathcal{B}_{\mu}[Z^*]+\mathcal{B}_{\mu}[Z^*]-Z^*\|_\infty\\
        &\leq\|\mathcal{B}_{\mu}[Z^*[\mu]]-\mathcal{B}_{\mu}[Z^*]\|_\infty+\|\mathcal{B}_{\mu}[Z^*]-Z^*\|_\infty\\        
        &\leq(1-(1-\gamma)\delta)\|Z^*[\mu]-Z^*\|_\infty+\|\mathcal{B}_{\mu}[Z^*]-Z^*\|_\infty\\            
        &=(1-(1-\gamma)\delta)\|Z^*[\mu]-Z^*\|_\infty+\|D'LZ^*\|_\infty\\            
        &\leq(1-(1-\gamma)\delta)\|Z^*[\mu]-Z^*\|_\infty+2D'\operatorname{maxdeg}(\mathcal{G})\|Z^*\|_\infty.            
    \end{align*}
\end{proof}

By using Lemma~\ref{lemma:mod-bellman-cont}, we can check the fast dynamics with a fixed slow component $\dot{Z}_t=\mathcal{T}[Z_t,\mu]$ satisfies one assumption of Theorem~\ref{thm:ODE-convergence} for a sufficiently small diffusion (relative) coefficient $D'$. 

\begin{proposition}
    \label{prop-GASE-V}
    For some $\delta>0$, we define $\tilde{\Delta}_{\delta,D',\mathcal{G}}(\mathcal{S})\coloneqq \left\{\mu\in\Delta(\mathcal{S})\relmiddle|\delta \leq\mu(s)\leq (1-D'\operatorname{maxdeg}(\mathcal{G})), \forall s\in\mathcal{S}\right\}$.
    For $\mu\in\tilde{\Delta}_{\delta,D',\mathcal{G}}(\mathcal{S})$, $\dot{Z}_t=\mathcal{T}[Z_t,\mu]$ has a unique globally asymptotically stable equilibrium (GASE) $Z^*[\mu]$. 
    Moreover, $Z^*[\mu]\in\mathbb{R}^{|\mathcal{S}|}$ is Lipschitz w.r.t. $\mu$ on $\tilde{\Delta}_{\delta,D',\mathcal{G}}(\mathcal{S})$. 
\end{proposition}

\begin{proof}
    The first claim follows by Lemma~\ref{lemma:mod-bellman-cont}.
    
    For $\mu,\mu'\in\tilde{\Delta}_{\delta,D',\mathcal{G}}(\mathcal{S})$,
    we have
    \begin{align*}
        \|Z^*[\mu]-Z^*[\mu']\|_\infty&=\|\mathcal{B}_{\mu}[Z^*[\mu]]-\mathcal{B}_{\mu'}[Z^*[\mu]]+\mathcal{B}_{\mu'}[Z^*[\mu]]-\mathcal{B}_{\mu'}[Z^*[\mu']]\|_\infty\\
        &\leq \|\mathcal{B}_{\mu}[Z^*[\mu]]-\mathcal{B}_{\mu'}[Z^*[\mu]]\|_\infty + \|\mathcal{B}_{\mu'}[Z^*[\mu]]-\mathcal{B}_{\mu'}[Z^*[\mu']]\|_\infty\\
        &\leq \|(\operatorname{diag}(\mu-\mu')\cdot Z^*[\mu]\|_\infty+ (1-(1-\gamma)\delta)\|Z^*[\mu]-Z^*[\mu']\|_\infty\\
        &< \|\mu-\mu'\|_\infty\|Z^*[\mu]\|_\infty+ (1-(1-\gamma)\delta)\|Z^*[\mu]-Z^*[\mu']\|_\infty\\
        &< \frac{(1-\gamma)\delta+2D'\operatorname{maxdeg}(\mathcal{G})}{(1-\gamma)^2\delta}\beta\|r\|_\infty\|\mu-\mu'\|_\infty\notag\\
        &+ (1-(1-\gamma)\delta)\|Z^*[\mu]-Z^*[\mu']\|_\infty
        ,
    \end{align*}
    where we used
    \begin{align*}
        \|Z^*[\mu]\|_\infty&\leq\|Z^*[\mu]-Z^*\|_\infty+\|Z^*\|_\infty\\
        &\leq \frac{(1-\gamma)\delta+2D'\operatorname{maxdeg}(\mathcal{G})}{(1-\gamma)\delta}\|Z^*\|_\infty\\
        &\leq \frac{(1-\gamma)\delta+2D'\operatorname{maxdeg}(\mathcal{G})}{(1-\gamma)^2\delta}\beta\|r\|_\infty.    
    \end{align*}
    Hence, we obtain
    \begin{equation}
        \|Z^*[\mu]-Z^*[\mu']\|_\infty< \frac{(1-\gamma)\delta+2D'\operatorname{maxdeg}(\mathcal{G})}{(1-\gamma)^3\delta^2}\beta\|r\|_\infty\|\mu-\mu'\|_\infty.
    \end{equation}
\end{proof}

We then check the slow dynamics with an equibrated fast component $\dot{\mu}_t=\mathcal{P}[Z^*[{\mu_t}],\mu_t]$ satisfies another assumption of Theorem~\ref{thm:ODE-convergence} for a sufficiently small transition rate $\varepsilon^p$. 
\begin{proposition}
    \label{prop-GASE-rho}
    Assume there exists $\delta\in(0,1)$ such that $(\mu_t)_{t\leq 0}$ generated by $\dot{\mu}_t=\mathcal{P}[Z^*[{\mu_t}],\mu_t]$ is confined inside $\tilde{\Delta}_{\delta,D',\mathcal{G}}(\mathcal{S})$ and 
    \[2|\mathcal{S}|\varepsilon^p\left(\exp(2M_{\beta,\gamma,r})+\frac{(1-\gamma)\delta+2D'\operatorname{maxdeg}(\mathcal{G})}{\gamma(1-\gamma)^2\delta^2}\frac{(M_{\beta,\gamma,r})^2}{\exp(m_{\beta,\gamma,r})}\right)<1.\]
    Then $\dot{\mu}_t=\mathcal{P}[Z^*[{\mu_t}],\mu_t]$ has a unique GASE $\mu^*$.
\end{proposition}
\begin{proof}
    $\mathcal{P}[Z,\mu]$ is contraction w.r.t. $\mu$:
    \begin{align*}
        &\|\mathcal{P}[Z^*[{\mu}],\mu]-\mathcal{P}[Z^*[{\mu'}],\mu']\|_\infty\\
        \leq &\|\mathcal{P}[Z^*[{\mu}],\mu]-\mathcal{P}[Z^*[{\mu}],\mu']\|_\infty+\|\mathcal{P}[Z^*[{\mu}],\mu']-\mathcal{P}[Z^*[{\mu'}],\mu']\|_\infty\\
        \leq&2\varepsilon^p\exp(2M_{\beta,\gamma,r})|\mathcal{S}|\|\mu-\mu'\|_\infty+\frac{2|\mathcal{S}|\varepsilon^p M_{\beta,\gamma,r}}{\exp(m_{\beta,\gamma,r})}\|Z^*[{\mu}]-Z^*[{\mu'}]\|_\infty\\
        <&2|\mathcal{S}|\varepsilon^p\left(\exp(2M_{\beta,\gamma,r})+\frac{(1-\gamma)\delta+2D'\operatorname{maxdeg}(\mathcal{G})}{(1-\gamma)^3\delta^2}\beta\|r\|_\infty\frac{M_{\beta,\gamma,r}}{\exp(m_{\beta,\gamma,r})}\right)\|\mu-\mu'\|_\infty\\
        =&2|\mathcal{S}|\varepsilon^p\left(\exp(2M_{\beta,\gamma,r})+\frac{(1-\gamma)\delta+2D'\operatorname{maxdeg}(\mathcal{G})}{\gamma(1-\gamma)^2\delta^2}\frac{(M_{\beta,\gamma,r})^2}{\exp(m_{\beta,\gamma,r})}\right)\|\mu-\mu'\|_\infty\\
        <&\|\mu-\mu'\|_\infty.
    \end{align*}
    Thus, there exists a unique GASE $\mu^*$ by contraction mapping theorem and $\mu_t$ converge to it.
\end{proof}

\section{Boundedness of our dynamics}
For linear production~\eqref{eq:dist-G-learning-diff}, we find a sufficient condition of reward which ensures the non-negativity of the endogenous cue concentration: $r(s)\geq 0$ and $Z_0(s)\geq 0$ for all $s\in\mathcal{S}$.

\begin{proposition}
    Let $Z_{\max}=\max\{\beta \max_{s\in\mathcal{S}} r(s),\max_{s\in\mathcal{S}} Z_0(s)\}$ and $Z_{\min}=\min\{\beta \min_{s\in\mathcal{S}} r(s),\min_{s\in\mathcal{S}} Z_0(s)\}$.
    
    For learning rate $\alpha$ s.t. $0\leq\alpha N\leq 1-D\operatorname{maxdeg}(\mathcal{G})$, we obtain upper and lower bounds of the iteration $Z_k$ ($k=0,1,\ldots$) generated by Eq.~\eqref{eq:dist-G-learning-diff}:
    \begin{subequations}
        \begin{align}
            Z_k(s)&\leq (1+\gamma+\cdots \gamma^k) Z_{\max},\label{eq:Giteration-up-bound}\\
            Z_k(s)&\geq Z_{\min},\label{eq:Giteration-lw-bound}
        \end{align}
    \end{subequations}
    for each $s\in\mathcal{S}$.
    Moreover, boundedness of the iteration follows from $\gamma\in[0,1)$.
\end{proposition}

\begin{proof}
    We show Ineq.~\eqref{eq:Giteration-up-bound} by induction.
    When $t=0$, it holds by definition.
    Assume it holds when $t=k$ for $k=0,1,\ldots$.
    Then, we obtain for $s\in\mathcal{S}$,
    \begin{align*}
        Z_{k+1}(s)&=Z_k(s)+\alpha\mu_t(s)(\beta r(s) + \log\sum_{s'\in\mathcal{S}}p(s'|s)\exp(\gamma Z_k(s'))-Z_k(s))+D\sum_{s'\in\mathcal{S}}L_{ss'}Z_k(s)\\
        &\leq(1-\alpha\mu_t(s)-D\operatorname{deg}(s))(1+\gamma+\gamma^2+\cdots+\gamma^k)Z_{\max}\notag\\&+\alpha\mu_t(s)(Z_{\max}+ \gamma (1+\gamma+\gamma^2+\cdots+\gamma^k)Z_{\max})+D\operatorname{deg}(s)(1+\gamma+\gamma^2+\cdots+\gamma^k)Z_{\max}\\
        &\leq (1+\gamma+\gamma^2+\cdots+\gamma^{k})Z_{\max}+\alpha\mu_t(s)\gamma^{k+1}Z_{\max}\\
        &\leq (1+\gamma+\gamma^2+\cdots+\gamma^{k})Z_{\max}+\alpha N\gamma^{k+1}Z_{\max}\\
        &\leq (1+\gamma+\gamma^2+\cdots+\gamma^{k+1})Z_{\max},
    \end{align*}
    which implies Ineq.~\eqref{eq:Giteration-up-bound} holds when $t=k+1$. 
    This completes the induction.
    
    We can show the lower bound~\eqref{eq:Giteration-lw-bound} similarly.
\end{proof}





\end{document}